\newif\ifsingle

\ifsingle
\documentclass[12pt,draftclsnofoot, onecolumn]{IEEEtran}		
\else		
\documentclass[journal]{IEEEtran}
\fi

\usepackage{amssymb,amsthm,pict2e,enumerate}
\usepackage{array}
\usepackage{booktabs}
\usepackage{color}
\usepackage{bm}
\usepackage{amsfonts}
\usepackage{acronym}

\usepackage{amsmath,dsfont}
\usepackage{cite}
\usepackage{url}
\usepackage{color}
\usepackage{soul}
\usepackage{algorithmicx,algorithm}

\ifCLASSINFOpdf
\usepackage[pdftex]{graphicx}
\else
\usepackage[dvips]{graphicx}
\fi

\ifCLASSOPTIONcompsoc
\usepackage[caption=false,font=normalsize,labelfon
t=sf,textfont=sf]{subfig}
\else
\usepackage[caption=false,font=footnotesize]{subfig}
\fi

\newtheorem{Thm}{Theorem}

\newtheorem{Prop}[Thm]{Proposition}

\newcommand{\E}{\mathds{E}}
\newcommand{\myVec}[1]{{\bm{#1}}}
\newcommand{\myMat}[1]{{\bm{#1}}}
\newcommand{\mySet}[1]{\mathcal{#1}}
\newcommand{\mvn}[4]{f_{G_{#4}}\big(#1;#2,#3\big)}
\newcommand{\ui}{{\bm{u}}}

\newcommand{\mean}[1]{{\bm{m}_{#1}}}
\newcommand{\var}[1]{{\bm{C}_{#1}}}

\ifsingle
\newcommand{\includefig}[1]{\includegraphics[width = 0.5\columnwidth]{#1} 	\vspace{-0.4cm}}
\else
\newcommand{\includefig}[1]{\includegraphics[width = 3 in]{#1} 	\vspace{-0.3cm}}
\fi 

\acrodef{majorcom}[MAJoRCom]{multi-carrier agile joint radar communication}
\acrodef{caesar}[CAESAR]{carrier agile phased array radar}
\acrodef{csi}[CSI]{channel state information}

\begin{document}

	%
	\title{MAJoRCom: A Dual-Function Radar Communication System Using Index Modulation}
	%
	%
	%
	
	\author{Tianyao Huang, Nir Shlezinger, Xingyu Xu, Yimin Liu, and Yonina C. Eldar
		\thanks{Parts of this work \cite{Huang2019} were accepted for presentation in the 2019 IEEE International Workshop on Signal Processing Advances in Wireless Communications (SPAWC), Cannes, France.
		This work received funding from the National Natural Science Foundation of China under Grants 61571260 and 61801258, from the European Union’s Horizon 2020 research and innovation program under grant No. 646804-ERC-COG-BNYQ, and from the Air Force Office of Scientific Research under grant No. FA9550-18-1-0208}
		\thanks{T. Huang,  X. Xu, and Y. Liu are with the EE Department, Tsinghua University, Beijing, China (e-mail:
			huangtianyao@tsinghua.edu.cn; xy-xu15@mails.tsinghua.edu.cn; yiminliu@tsinghua.edu.cn).}
		\thanks{N. Shlezinger and Y. C. Eldar are with the Faculty of Math and CS, Weizmann Institute of Science, Rehovot,  Israel (e-mail: nirshlezinger1@gmail.com; yonina.eldar@weizmann.ac.il).} 
		\vspace{-1.0cm}
	}

	\maketitle
	\begin{abstract} 
		Dual-function radar communication (DFRC) systems implement both sensing and communication using the same hardware. 
		Such schemes are often more efficient in terms of size, power, and cost, over using distinct radar and communication systems.
		Since these functionalities share resources such as spectrum, power, and antennas, DFRC methods typically entail some degradation in both radar and communication performance. 
		In this work we propose a  DFRC scheme based on the \ac{caesar}, which combines frequency and spatial agility. The proposed DFRC system,  referred to as \ac{majorcom}, exploits the inherent spatial and spectral randomness of \ac{caesar} to convey digital messages in the form of index modulation. The resulting communication scheme naturally coexists with the radar functionality, and thus does not come at the cost of reduced radar performance.  
		We analyze the performance of \ac{majorcom}, quantifying its achievable bit rate. In addition, we develop a low complexity decoder and a codebook design approach, which simplify the recovery of the communicated bits. 
		Our numerical results demonstrate that \ac{majorcom} is capable of achieving a bit rate which is comparable to utilizing independent communication modules without affecting the radar performance, and that our proposed low-complexity decoder allows the receiver to reliably recover the transmitted symbols with an affordable computational burden. 
	\end{abstract}
	

	%
	\IEEEpeerreviewmaketitle
	
	\acresetall
	\vspace{-0.2cm}
	\section{Introduction}
	\vspace{-0.1cm}
	Recent years have witnessed a growing interest in dual-function radar communication (DFRC) systems. 
	Many practical applications, including autonomous vehicles, commercial flight control, and military  radar systems, implement both sensing as well as communications \cite{Hassanien2016, Paul2017, mishra2019towards,ma2019joint}. 
	Jointly implementing radar and communication contributes to reducing the number of antennas \cite{Tavik2005},  system size, weight, and power consumption \cite{Liu2017a}, as well as alleviating concerns for electromagnetic compatibility (EMC) and spectrum congestion issues \cite{Hassanien2016}.
	In one of the most common models for joint radar and communications, the DFRC system acts as the radar transceiver and communications transmitter simultaneously. This setup, which is considered henceforth, is commonly referred to as the {\em monostatic broadcast channel} \cite[Sec. III-C]{Paul2017}.  In such scenarios, radar is  regarded as the primary function and  communications as the secondary one, sharing the high power and large bandwidth of the radar \cite{Zheng2018, Wang2018}. 
	
	Since DRFC systems implement both radar and communications using a single hardware device, these functionalities inherently share some of the system resources, such as spectrum, antennas, and power. To facilitate their coexistence, many different DFRC approaches have been proposed in the literature.  
	In a single antenna radar or traditional phased array radar that transmits a single waveform, a common scheme is to utilize the communication signal as the radar probing waveform \cite{Sturm2011}. 
	Such dual-function waveforms include  phase modulation, as well as orthogonal frequency division multiplexing (OFDM) signaling \cite{Hassanien2016a,Sturm2011}. The design of such waveforms to fit a  given beam pattern was studied in \cite{Liu2017b}. However, this approach tends to come at the cost of reducing radar performance compared to using dedicated radar signals  \cite{Wang2018, Liu2018}. Furthermore, transmitting non-constant modulus communication waveforms may result in low power efficiency when using practical non-linear amplifiers. 
	
	Another common DFRC approach is to utilize different signals for radar and communications, designing the functionalities to co-exist by mitigating their cross interference. Multiple-input multiple-output (MIMO) radar systems in which a subset of the antenna array is allocated to radar and the rest to communications were studied in \cite{Liu2018}, along with the setup in which both functionalities utilize all the antennas. 
	Methods for treating the effect of spectrally interfering separate radar and communication systems were studied in \cite{Mahal2017, Liu2018a}, while \cite{nartasilpa2018communications} analyzed the effect of radar interference on communication systems. Frequency allocation among radar and communications was considered in \cite{bicua2019multicarrier}.
	Coexistence in MIMO DFRC systems can be realized using beamforming, namely, by generating multiple beams with different waveforms towards radar targets and communication users at diverse directions \cite{McCormick2017,Liu2018b}.
	The work \cite{Ma2018} proposed a scheme based on generalized spatial modulation (GSM) \cite{Wang2012}, in which some of the information bits are  conveyed in the selection of the antennas utilized for communication. 
	The drawback of these previous DFRC methods,  particularly when radar is the primary functionality, is that communication interferes with the radar, either via spectral interference, power sharing, or by reducing the number of available antennas, resulting in an inherent tradeoff between radar and communication performance \cite{Chiriyath2017, Qian2018}.
	
	An alternative DFRC strategy is to incorporate communication functionality into existing radar schemes. A common radar technique which can be extended into a DFRC system is  MIMO radar, in which each antenna element transmits a different orthogonal waveform, enhancing the flexibility in transmit beam pattern design \cite{Cohen2018}. The resulting waveform diversity can be exploited to embed information bits into the transmitted signal with minimal effect on the radar performance. For example, the information bits can be conveyed in the sidelobe levels \cite{Hassanien2016b} or via frequency hopping codes  \cite{Hassanien2017}.  The recent work \cite{Wang2018} studied permutation of antenna elements, each transmitting a different predefined orthogonal waveform, as a method for embedding information bits. However, since radar returns of all orthogonal waveforms are received by each antenna element, MIMO radar receivers usually operate at a large bandwidth, resulting in high complexity in hardware and computing. Consequently, these DFRC approaches may be difficult to implement in practice and cannot be applied in many existing radar architectures.

	In our previous work \cite{Huang2019b} we proposed \ac{caesar}, which is a radar scheme capable of approaching wideband performance while utilizing narrowband signals. This improved performance is achieved by combining the concept of frequency agile radar (FAR), in which the carrier frequencies vary from pulse to pulse \cite{Axelsson2007}, with {\em spatial agility}. In particular, \ac{caesar} randomly  chooses multiple frequencies simultaneously in a single pulse, and then selects a set of antennas for each chosen frequency such that each set of antennas uses a different frequency as depicted in Fig. \ref{fig:multi-carrier}. In the reception stage, each array element acquires the radar returns at the same single frequency  as in the transmitting stage, which reduces hardware complexity in comparison with MIMO radar architectures. The resulting radar scheme has excellent electronic counter-countermeasures (ECCM) and EMC performance; it supports spectrum sharing in congested electromagnetic environments; and its radar performance is comparable to that of costly wideband  radar \cite{Huang2019b}. 	 
	In addition to the aforementioned advantages, the inherent spectral and spatial randomness of \ac{caesar} can be utilized to convey information using index modulation methods, in which the indices of the building blocks (e.g., frequencies and/or antennas) are used to convey additional information bits \cite{Basar2016}, 
	{\em without degrading radar performance}. The resulting \ac{majorcom} system is the focus of the current work.

	\begin{figure}
		\centering
		\includegraphics[width=2.5in]{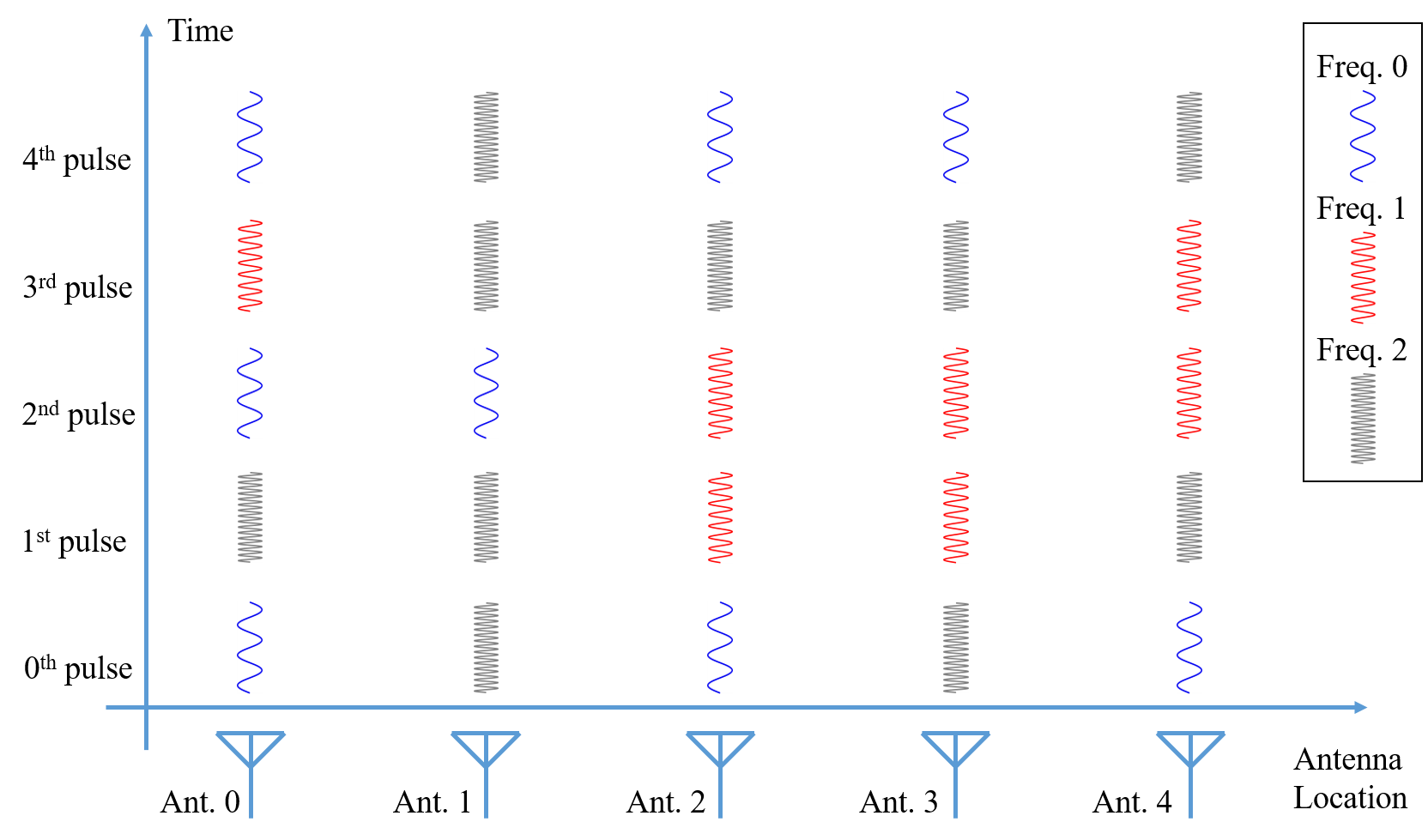}
		\vspace{-0.2cm}
		\caption{Transmission example of \ac{caesar} \cite{Huang2019b}. In every pulse of this example, two out of three carrier frequencies are emitted by different sub-arrays. For example, frequency 0 and 2 are selected in the 0-th pulse and are sent by antennas 0, 2, 4 and antennas 1, 3, respectively.  FAR is a special case of CAESAR, with only one out of three frequencies sent in each pulse.}
		\label{fig:multi-carrier}
		\vspace{-.6cm}
	\end{figure}	
	
	Here, we propose \ac{majorcom}: a DFRC system equipped with a phased array antenna, in which radar is the primary user and is based on \ac{caesar}. We show how \ac{caesar} is capable of conveying information to a remote receiver using index modulation.  
	\ac{majorcom} utilizes the selections of carrier frequencies and their allocation among the antenna elements of \ac{caesar} to convey digital information in a combination of frequency index modulation \cite{Basar2015} and spatial index modulation \cite{Basar2016}. Unlike previously proposed DFRC systems \cite{Hassanien2016a,Sturm2011,Ma2018,Mahal2017, Liu2017b,Liu2018,bicua2019multicarrier}, which use dedicated independent waveforms and/or antennas for communication, in \ac{majorcom} the ability to convey information is an inherent byproduct of the radar scheme. Consequently, communication transmission is naturally obtained from the radar design, and both functionalities coexist without cross interference. 

	We  analyze the communication performance of \ac{majorcom}. Since the communication functionality does not interfere with the radar subsystem, the radar performance of \ac{majorcom} is the same as \ac{caesar}, and was studied in our previous work \cite{Huang2019b}.
	Here, we first detail the scheme for embedding digital communication messages into the radar transmission. 
	We characterize the achievable rate of \ac{majorcom}, and show that the maximal number of bits which can be conveyed in each pulse  grows linearly with the number of transmit antennas and logarithmically with the number of available carrier frequencies.
	%
	%
	To overcome the increased computational complexity associated with index modulation decoding \cite{Wen2017}, we propose a low complexity communication receiver structure and design a permutation codebook to facilitate decoding. 
	\ac{majorcom} is evaluated in a numerical study, demonstrating its capability to achieve comparable communication rates with DFRC systems using antennas that are dedicated for communication only, without affecting the radar performance and resources.
	 
	Our main contributions are summarized as follows:
	\begin{itemize}
		\item We propose \ac{majorcom} which is a  DFRC system that arises from \ac{caesar}. The proposed communication scheme is based on frequency and spatial index modulation, in which selections of frequencies and the corresponding antenna elements are used to embed information,  {without requiring the transmitter to have \ac{csi}.} 
		These communication methods are inherent to the radar scheme, and thus do not affect the power and waveform of the radar functionality.
		\item We analyze the achievable information rate of \ac{majorcom}. In particular, we show that the maximal number of bits which can embedded into each pulse, representing an upper bound on the information rate which is achievable in high signal-to-noise ratio (SNR), grows logarithmically with the number of carrier frequencies. This indicates that increasing the agility of the radar also contributes to its achievable rate. 
		\item We propose a low complexity decoder for the proposed  scheme, which achieves comparable bit error rate (BER) performance as the optimal decoder. Codeword design approaches are also proposed to further facilitate decoding, at the cost of reducing the information rate.
	\end{itemize}
	The main advantage of \ac{majorcom} over previously proposed DFRC systems, e.g., \cite{bicua2019multicarrier,Hassanien2016a,Sturm2011,Ma2018, Liu2017b,Liu2018}, is that it provides the ability to communicate without affecting the radar subsystem, while supporting the usage of simple narrowband transceivers.

	The rest of paper is organized as follows. Section~\ref{sec:system} reviews \ac{caesar} and introduces \ac{majorcom}, which applies frequency selection and spatial permutation to convey digital messages. Section~\ref{sec:comm} is devoted to communication analysis, while Section~\ref{sec:complexity} introduces low-complexity receiver and codebook design methods. Numerical results are provided in Section~\ref{sec:sim}, followed by concluding remarks in Section~\ref{sec:conclusion}.
	 
	Throughout the paper we use the following notation: The sets $\mathbb{C}$, $\mathbb{R}$ and $\mathbb{Z}$ are the complex, real and integer numbers, respectively. 
	We use $| \cdot |$  for the magnitude or cardinality of a scalar value or a set, respectively. We denote by $\lfloor x \rfloor$ the largest  integer less than or equal to $x \in \mathbb{R}$. Uppercase and lowercase boldface letters are used for matrices and vectors, respectively.
	The $m$,$n$-th ($n$-th) element of matrix $\bm A$ (vector $\bm a$) is written as $[{\bm A}]_{m,n}$ ($[{\bm a}]_{n}$). 
	We use $\bm 0/\bm 1_{n \times m}$ to denote a $n \times m$ dimensional matrix with all entries being 0/1.
	The complex conjugate operator, transpose operator, and the complex conjugate-transpose operator are denoted by $(\cdot)^*$, $(\cdot)^T$, and $(\cdot)^H$. 
	We use $\| \cdot \|_p$ as the $\ell_p$ norm of an argument, and ${ \E}[\cdot]$ is the stochastic expectation.
	

	\vspace{-0.2cm}
	\section{\ac{majorcom} System Model}
	\label{sec:system}
	\vspace{-0.1cm}
	In this work, we propose \ac{majorcom}, which jointly implements radar  as well as the ability of communicating information to a remote receiver.  Radar is considered to be the primary user, and is based on the recently proposed \ac{caesar} scheme \cite{Huang2019b}. The communication method is integrated into \ac{caesar} to avoid coexistence issues. In order to formulate \ac{majorcom}, we first review \ac{caesar} in Subsection~\ref{subsec:multi-carrier}, after which we present its extension to a DFRC system in Subsection~\ref{sec:info}.

	\vspace{-0.2cm}
	\subsection{Carrier Agile Phased Array Radar}
	\label{subsec:multi-carrier}
	\vspace{-0.1cm} 
	\ac{caesar} is a recently proposed radar scheme \cite{Huang2019b}  which extends the concept of FAR \cite{Axelsson2007}. This technique was shown to  enhance the ECCM and EMC radar measures as well as achieve improved target reconstruction performance while avoiding costly instantaneous wideband components \cite{Huang2019b}.
	Broadly speaking, \ac{caesar} randomly changes the carrier frequencies from pulse to pulse, maintaining the frequency agility of FAR, while allocating these frequencies among its antenna elements in a random fashion, introducing spatial agility.  
	An illustration of this scheme is depicted in Fig. \ref{fig:multi-carrier}.

	To properly formulate \ac{caesar}, consider a radar system equipped with $L_{\rm R}$ antenna elements,  uniformly spaced with distance $d$ between two adjacent elements.  
	Let  $\mathcal{F}$ be the set containing the available carrier frequencies of cardinality $M$,  given by 
	 \begin{equation}
	 \mathcal{F}:=\{ f_c + m\Delta f | m \in \mySet{M}\},
	 \end{equation}
	 where $\mySet{M} := \{0,1,\dots,M-1\}$, $f_c$ is the initial carrier frequency, and $\Delta f$ is the frequency step.  Let $N$ be the number of radar pulses transmitted in each coherent processing interval, and $f_{n} \in \mathcal{F}$  denote the carrier frequency of the $n$-th pulse. Radar pulses are repeatedly transmitted, starting from time instance $nT_r$ to $nT_r + T_p$, $n\! \in\! \{0,1,\dots,N\!-\!1\} \!:=\! \mySet{N}$, where $T_r$ and $T_p$ are the pulse repetition interval and  duration, respectively,  $T_r > T_p$.

	In the $n$-th pulse, \ac{caesar} randomly selects a set of carrier frequencies $\mySet{F}_n$ from $\mathcal{F}$, $\mySet{F}_n \subset \mathcal{F}$. We assume that the cardinality of $\mySet{F}_n$ is constant, i.e., $|\mySet{F}_n| = K$ for each $n \in  \mySet{N}$, and write the elements of this set as $\mySet{F}_n = \{\Omega_{n,0},\dots,\Omega_{n,K-1} \}$.
	A sub-array is allocated for each frequency, such that all the antenna array elements are utilized for transmission and each element transmits at a single carrier frequency.
	Denote by $f_{n,l} \in \mySet{F}_n$ the frequency used by the $l$-th antenna array element, i.e., if $\Omega_{n,k}$ is the frequency used by the $l$-th element then $f_{n,l} = \Omega_{n,k}$. The waveform sent from the $l$th element for the $n$-th pulse is expressed as $\phi(f_{n,l}, t-nT_r)$, where $\phi(f, t) :=  {\rm rect}\left({t}/{T_p} \right) e^{j2\pi f t}$. 
	In order to direct the antenna beam pointing towards a desired angle $\theta$, the signal transmitted by each antenna is weighted by the function $w_l(\theta,f_{n,l})\in \mathbb{C}$, which is set to \cite{Pillai1989}
	\begin{equation}
	\label{eq:weight}
	w_l(\theta,f_{n,l}) = e^{j2\pi f_{n,l} ld \sin \theta/c},
	\end{equation}
	where $c$ denotes the speed of light. 
	The transmission of the $l$-th array element can thus be written as
	\begin{equation}
	\label{eq:x}
	\left[ \bm x(n,t) \right]_l= w_l(\theta,f_{n,l}) \phi(f_{n,l}, t-nT_r).
	\end{equation}
	
	The vector $\bm x(n,t) \in \mathbb{C}^{L_{\rm R}}$ in \eqref{eq:x} denotes the transmission vector of the full array for the $n$-th pulse at time instance $t$. An illustration of such a transmission is depicted in Fig. \ref{fig:multi-carrier}.   
	The transmitted signal \eqref{eq:x} can also be expressed by grouping the array elements which use the same frequency $\Omega_{n,k}$, $k = 0,\dots,K-1$. Let ${\bm x}_k(n,t)\in \mathbb{C}^{L_{\rm R}}$ represent the portion of ${\bm x}(n,t)$, which utilizes $\Omega_{n,k}$, i.e., $\bm x(n,t) = \sum _{k=0}^{K-1}\bm x_k(n,t)$. The transmitted signal can now  be written as
	\begin{equation}
	\bm x(n,t) 
	: = \sum \limits_{k=0}^{K-1} \bm P(n,k) \bm w\left(\theta,\Omega_{n,k}\right) \phi\left(\Omega_{n,k}, t-nT_r\right),
	\label{eqn:TxSignal}
	\end{equation}
	where $\bm P(n,k) \in \{0,1\}^{L_{\rm R} \times L_{\rm R}}$ is a diagonal selection matrix with diagonal $\bm p(n,k)  \in \{0,1\}^{L_{\rm R}}$, whose  $l$-th entry is 1 if the corresponding array element uses $\Omega_{n,k}$ and 0 otherwise, i.e., $\left[ \bm P(n,k) \right]_{l,l} =\left[ \bm p(n,k) \right]_{l} = 1$ when $\left[ \bm x_k(n,t) \right]_{l} \neq 0$.
	 
	In the reception stage of the $n$-th pulse, i.e., $nT_r + T_p < t< (n+1)T_r$, the $l$-th antenna element only receives radar returns at frequency $f_{n,l}$, and abandons returns at other frequencies, facilitating the usage of narrowband radar receiver and simplifying the hardware requirements. Our proposed extension of \ac{caesar} to a DFRC system, detailed in the following subsection, exploits the transmitted signal model \eqref{eqn:TxSignal}, and does not depend on the observed radar returns and processing strategy.  The readers are referred to \cite{Huang2019b} for a detailed description of the received radar signal model, target recovery methods, and radar performance analysis of \ac{caesar}.

	
	\vspace{-0.2cm}
	\subsection{Information Embedding Scheme}
	\label{sec:info} 
	\vspace{-0.1cm}
	The inherent randomness in the selection of carrier frequencies and their allocation among the transmit antennas can be exploited to convey information in the form of index and permutation modulations. Index modulation refers to the embedding of information bits through indices of certain parameters involved in the transmission \cite{Basar2016}, most commonly the subcarrier index in OFDM modulation, i.e., frequency index modulation \cite{Basar2015}, or the antenna selection in MIMO communications, namely, spatial modulation \cite{Wang2012}. 
	 \ac{caesar} randomly selects an index corresponding to a set of carrier frequencies, and permutes the selected frequencies and the corresponding antenna elements, which can either be treated as an index of a specific permutation, or as a permutation modulation codeword \cite{Slepian1965}. 
	By doing so, \ac{caesar} realizes a DFRC system, as illustrated in Fig. \ref{fig:system} for the  setting of $|\mathcal{F}| = M = 2$.
	 Consequently, a natural extension of \ac{caesar} is to utilize this randomness to convey information to a remote receiver, thus realizing digital communications without affecting the radar functionality. 
	 
 	\begin{figure}
 		\centering
 		\includegraphics[width=3.5in]{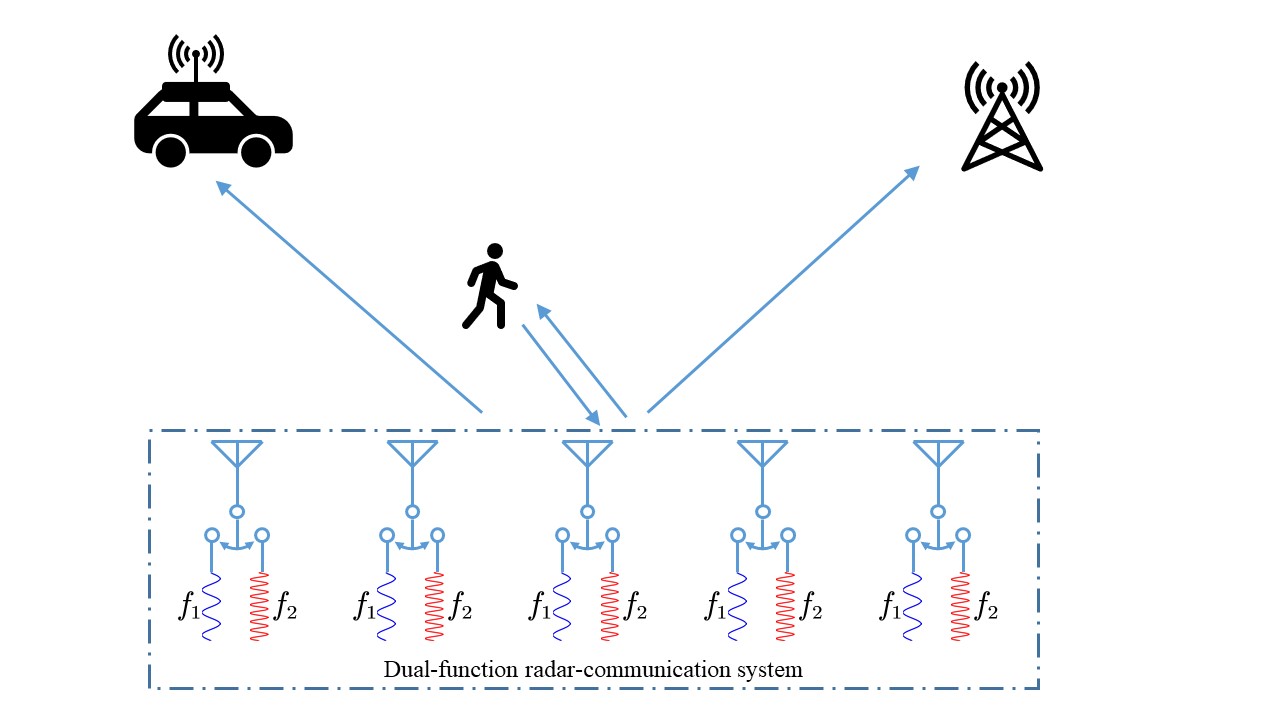}
 		\caption{A phased array DFRC system, which can detect targets (e.g., the pedestrian) and send communication symbols to remote receivers. Each array element independently selects the carrier frequency, e.g. from $f_1$ and $f_2$.}
 		\label{fig:system}
 		\vspace{-0.6cm}
 	\end{figure}
	 
	The proposed information embedding method is applied identically on each pulse, where transmitting more pulses results in more bits being conveyed to the receiver. Consequently, in order to formulate the embedding method, we only consider a single pulse in this section. Accordingly, we simplify our notations as follows: $\mySet{F}: = \mySet{F}_n$, $\bm P_k :=\bm P(n,k)$, $\bm p_k :=\bm p(n,k)$, $\bm x(t) :=\bm x(n,t)$, $\bm x_k(t):=\bm x_k(n,t)$, and $\bm w_k := \bm w\left(\theta,\Omega_{n,k}\right)$. 
	
	Before transmitting the dual function waveform, \ac{caesar} first selects frequencies and then allocates array elements to each frequency. The randomness of digital communication messages is utilized to convey information in the selection of the frequencies subset and in the allocation of the subset among the transmit antennas. We propose to exploit this fact to generate  two sets of codewords, combined into a hybrid modulation strategy, as discussed next.
	
	\subsubsection{Frequency Index Modulation}
	Recall that at each transmission, $K$ out of $M$  frequencies in $\mathcal{F}$ are used. The set of possible frequency selections at each pulse is denoted by
	\begin{equation}
	\mho :=\left\{ \mySet{F}^{(i)}\left| \left|\mySet{F}^{(i)}\right|=K,\mySet{F}^{(i)}\subset \mathcal{F}, i = 0,1,2,\dots \right. \right\},
	\end{equation} 
	where the superscript $(i)$ stands for the $i$-th codeword in the set $\mho$. The number of  possible frequency selections is thus 
	\begin{equation}
	\left| \mho \right| = {{M}\choose{K}} = \frac{M!}{K!(M-K)!}.
	\end{equation}
	
	\subsubsection{Spatial Index Modulation}
	Once the carrier frequencies are selected, each antenna element uses a single frequency  to transmit its monotone waveform. To mathematically formulate this allocation, we define $L_K:=L_{\rm R}/K \geq 1$, which is assumed to be an integer, and allow each frequency to be utilized by exactly $L_K$ antenna elements\footnote{The assumption that $L_R / K$ is an integer is used only to facilitate the formulation of the permutation technique. Clearly, the proposed spatial index modulation can be extended to the case that $L_{\rm R}$ is not an integer multiple of $K$ and that antennas are unevenly allocated by adapting the above arguments.} assigned to the selected $K$ frequencies.  The diagonal selection matrices $\{\bm P_k\}$ uniquely describe the allocation of antenna elements. We note that ${\rm tr}\left( \bm P_k\right) = L_K$, as exactly $L_K$ antennas use the $k$-th frequency, and $\sum _{k=0}^{K-1}\bm P_k\ = \bm I_{L_{\rm R}}$, indicating that all the antenna elements are utilized. Let $\mathcal{P}$ denote the set of all possible allocation patterns, given by   
	\begin{equation}
	\mathcal{P} :=\left\{\left. \bm P^{(i)}_0,\dots, \bm P^{(i)}_{K-1} \right|  i=0,1,\dots \right\},
	\end{equation}
	where the superscript $(i)$ stands for the $i$-th allocation pattern. Note that the number of   patterns is 
	\begin{equation}
	\left| \mathcal{P}\right| = \frac{L_{\rm R}!}{\left(L_K!\right)^{K}}.
	\end{equation}
	 
	As an example, consider a \ac{majorcom} system equipped with $L_{\rm R}=4 $ antennas, transmitting $K=2$ frequencies in each pulse, namely, each frequency is utilized by $L_K = 2$ antennas. In this case, the number of codewords which can be conveyed by this spatial permutation is $\frac{4!}{\left(2!\right)^{2}} = 6$. The first three possible selection patterns are:
	\begin{equation}
	\label{eq:P_example}
	\begin{split}
	\bm p_0^{\left( 0 \right)}=\left[ 1,1,0,0 \right]^T &,
	\bm p_1^{\left( 0 \right)} =\left[ 0,0,1,1 \right]^T ,\\
	\bm p_0^{\left( 1 \right)} =\left[ 1,0,1,0 \right]^T &,
	\bm p_1^{\left( 1 \right)} =\left[ 0,1,0,1 \right]^T ,\\
	\bm p_0^{\left( 2 \right)}=\left[ 1,0,0,1 \right]^T &,
	\bm p_1^{\left( 2 \right)} =\left[ 0,1,1,0 \right]^T.
	\end{split}
	\end{equation}
	The remaining three matrices are obtained by interchanging the subscripts, e.g., by setting $\bm p_0^{\left( 3 \right)} = \bm p_1^{\left( 0 \right)}$, $\bm p_1^{\left( 3 \right)} = \bm p_0^{\left( 0 \right)}$.
	
	\subsubsection{Hybrid modulation}
	Combining frequency and antenna selection yields a hybrid frequency and spatial index modulation scheme, in which  the total number of codewords is
	\begin{equation}
	\left| \mho \right| \left| \mathcal{P}\right| = \frac{M!}{K!(M-K)!} \frac{L_{\rm R}!}{\left(L_K!\right)^{K}}.
	\label{eqn:NumOfBits}
	\end{equation}
	It follows from \eqref{eqn:NumOfBits} that the maximum number of bits which can be conveyed in each pulse is 
	\begin{equation}
	\log_2\left| \mho \right| \!+\! \log_2\left| \mathcal{P}\right|\!\! =\!\! \log_2 \frac{M!}{K!(M-K)!} \!+ \! \log_2 \frac{L_{\rm R}!}{\left(L_K!\right)^{K}}. 
	\label{eq:approximate_bits}
	\end{equation}
	Using Stirling's formula $\log_2 n! \approx n \log_2 n - n \log_2 e$, 
	the number of bits \eqref{eq:approximate_bits} can be approximated as
	\begin{align}
	\log_2\left| \mho \right| \!+\! \log_2\left| \mathcal{P}\right| 
	&\approx \log_2 \left(\frac{M^M}{(M\!-\!K)^{M\!-\!K}K^K}\right) \!+\!L_{\rm R}\log_2K  \notag \\
	&\approx K\log_2 M+L_{\rm R}\log_2K.\label{eq:approximate_bits2}
	\end{align}
	This approximation holds for a large number of antennas $L_{\rm R}$ and a large number of frequencies $M$ such that  $L_{\rm R} \gg K$ and $M \gg K$. It follows from \eqref{eq:approximate_bits2} that the number of bits grows linearly with  $L_{\rm R}$ and logarithmically with $M$, indicating the theoretical benefits of utilizing  \ac{majorcom} with large-scale antenna arrays where $M$ is large.
	 
	The proposed information embedding scheme is carried out as follows: At each pulse, the input bits are divided into two sets. The first set of bits is used for selecting the frequencies $\mySet{F}$ from $\mho$, while the remaining bits determine the pattern of antenna allocation from $\mathcal{P}$. An example of this scheme is depicted in Fig. \ref{fig:PM_signaling}.
	\begin{figure}
		\centering
		\includegraphics[width=2.5in]{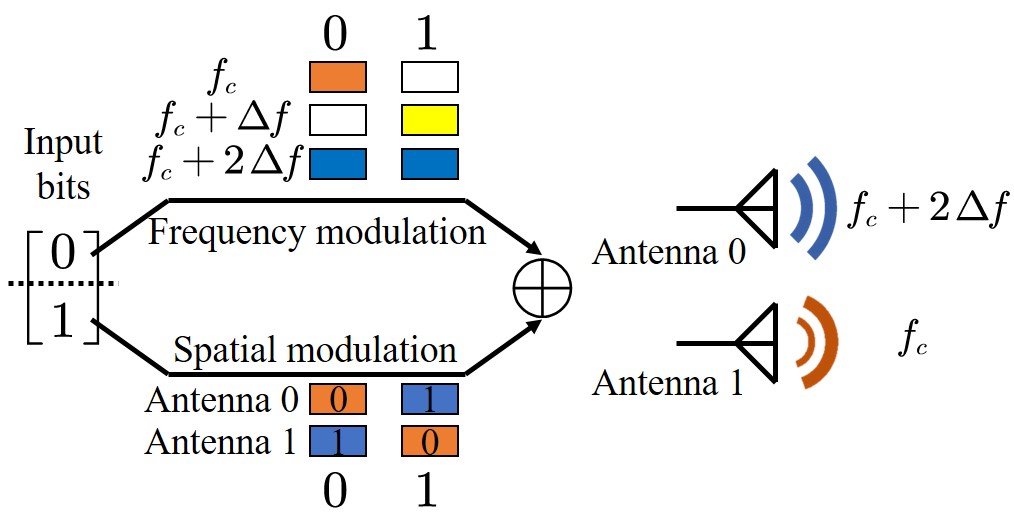}
		\vspace{-0.2cm}
		\caption{Hybrid frequency and spatial index signaling of \ac{majorcom}. }
		\label{fig:PM_signaling}
		\vspace{-0.6cm}
	\end{figure}
	%
	%
	This method bears some similarity to  generalized space-frequency index modulation proposed in \cite{Datta2016}. In particular, both schemes convey information in the selection of the carrier frequencies as well as in the form of the signal transmitted by each  antenna element. Nonetheless, while  \cite{Datta2016} transmits an OFDM signal consisting of multiple subcarriers from a subset of the complete antenna array, \ac{majorcom} utilizes a single carrier frequency at each transmit antenna and transmits a radar waveform using all the available antennas. Consequently, our approach transmits constant modulus monotone signals, and utilizes the complete antenna array,  maximizing the radar power and aperture. For the radar function, the use of complete antenna array is important, because it leads to a more directional beam and higher antenna gain, which is more suitable for target detection, especially in tracking mode \cite{Eli2013}. In contrast,  \cite{Datta2016} embeds information in the selection of active antennas,  
	leading to incomplete antenna aperture and reduction of radar performance.

 \ac{majorcom} does not require the DFRC system to have \ac{csi}, namely, no a-priori knowledge of the channel to the receiver is required in order to embed the information, as opposed to, e.g., spatial beamforming-based DFRC systems \cite{Liu2018b, Sodagari2012AProjection}.  
 Such knowledge is only needed at the receiver to facilitate decoding, as discussed in the following section. 
Furthermore, while we assume that the radar waveform $\phi(f, t) :=  {\rm rect}\left({t}/{T_p} \right) e^{j2\pi f t}$ does not convey informative bits, 
\ac{majorcom} can clearly be extended to embed data into the waveform. For example, by utilizing GSM \cite{Wang2012}, the proposed hybrid frequency and spatial modulation can potentially increase the communication rate.  
However, such a modification would come at the cost of some degradation in radar performance as the radar scheme depends on the waveform and available resources.  We leave this investigation to future work.
	
	
	\vspace{-0.2cm}
	\section{Communication Performance Analysis}
	\label{sec:comm}
	\vspace{-0.1cm}
	We now analyze the communication performance of \ac{majorcom} in terms of achievable rate. To that aim, we first derive the received communication signal model in Subsection~\ref{subsec:comm_signal_model}, and then characterize the achievable rate in Subsection~\ref{subsec:rate}. 
	This analysis allows us to numerically evaluate the communication capabilities of \ac{majorcom} in Section \ref{sec:sim}, where we demonstrate that its achievable rate is comparable to using dedicated communication waveforms, without affecting radar performance.
	
	\vspace{-0.2cm}
	\subsection{Received Communication Signal Model}
	\label{subsec:comm_signal_model}
	\vspace{-0.1cm}
	To model the signal observed by the remote communication receiver, let $L_{\rm C}$ denote the number of receiver antennas, and consider a memoryless additive white Gaussian noise channel. The channel output observed by the receiver, ${\bm y}_{\rm C}(t) \in \mathbb{C}^{L_{\rm C}}$, is given by
	\vspace{-0.1cm}
	\begin{equation}
	\label{equ:yCnt}
	{\bm y}_{\rm C}(t) = \sum \limits_{k=0}^{K-1}\bm H \bm x_k(t)  + \bm n_{\rm C}(t),
	\vspace{-0.1cm}
	\end{equation}
	where  ${\bm n}_{\rm C}(t) \in \mathbb{C}^{L_{\rm C}}$ is the additive Gaussian noise signal and $\bm H \in \mathbb{C}^{L_{\rm C} \times L_{\rm R}}$ is the channel matrix representing the complex-valued fluctuations between the \ac{majorcom} system and the remote receiver. 
	The proposed model can be extended to account for frequency selective channels by using bandlimited waveforms whose bandwidth is no larger than the channel coherence bandwidth. In this case, the  matrix $\bm H $ in \eqref{equ:yCnt} is replaced with the frequency index dependent matrix ${\bm H}_k$. 
	 
	After down-conversion by $e^{-j2\pi f_c t}$,  the receiver samples the signal at time instances $iT_s$, where $T_s$ is the sampling interval, and $i = 0,1,\dots,\lfloor T_p/T_s \rfloor$, resulting in $L_T := \lfloor T_p/T_s \rfloor + 1$ outputs per pulse. We assume that the receiver observes the complete frequency range $\mathcal{F}$, and applies Nyquist sampling rate of the entire bands, $T_s = \frac{1}{M\Delta f}$. We refer to \cite{Xi2014, Ioushua2017} and references therein for sub-sampling approaches. 
	By letting $\bm Y_{\rm C}, \bm N_{\rm C} \in \mathbb{C}^{L_{\rm C} \times L_T}$ denote the sampled channel output and noise corresponding to a single pulse in matrix form, respectively, it follows from the transmit signal model \eqref{eqn:TxSignal} that
	\vspace{-0.1cm}
	\begin{equation}
	\bm Y_{\rm C} =  \sum \limits_{k=0}^{K-1} \bm H \bm P_k \bm w_k {\bm \psi}_{c_{k}}^T + \bm N_{\rm C}.
	\label{eq:YC}
	\vspace{-0.1cm}
	\end{equation}
	In \eqref{eq:YC}, we define $c_k := \left( \Omega_{n,k} -f_c \right)/\Delta f \in \mySet{M}$ as the frequency codeword corresponding to $\Omega_{n,k}$, and ${\bm \psi}_{c_k} := \big[1,e^{j2\pi c_k\Delta f T_s \cdot 1},\dots, e^{j2\pi c_k\Delta f T_s \cdot (L_T-1)}\big]^T\in \mathbb{C}^{L_T}$ as the baseband signal corresponding to the frequency codeword $c_k$.

	We assume that the receiver knows the number of frequencies $K$, the  steering vectors $\{\bm w_k\}$, and has \ac{csi}, i.e., knowledge of the channel matrix $\bm H$, and the distribution of the additive noise. Recall that such \ac{csi} is only required at the receiver side. The fact that for a fixed frequency-antenna allocation, the transmitted waveform is deterministic, can be utilized to facilitate channel acquisition  in a pilot-aided fashion when $\bm H$ has to estimated. We leave the analysis of channel estimation and its effect on the system performance, as well as the design of frequency-antenna allocation pilot sequences for future investigation, and focus here on the case where $\bm H$ is known at the receiver.  Under the above signal model, we next study the achievable rate.
	
	\vspace{-0.2cm}
	\subsection{Achievable Rate Analysis}
	\label{subsec:rate}
	\vspace{-0.1cm}
	In order to evaluate the proposed communication scheme, we characterize its achievable rate, namely, the maximal number of bits which can be reliably conveyed to the receiver at a given noise level in each pulse. To facilitate the analysis, we assume that each discrete-time channel output represents a single pulse, i.e., $L_T = 1$. It is emphasized that the following analysis can also be extended to any positive integer value of $L_T$. Under this model, for each pulse,  the input-output relationship of the communication channel \eqref{eq:YC} is given by
	\begin{equation}
	\label{RxModel1}
	{{\bm{y}}_{\rm C}}= \bm{H}{\bm{x}} + {{\bm{n}}_{\rm C}},
	\end{equation}
	where ${\bm{x}}  = \sum_{k = 0}^{K - 1} {{{\bm{P}}_k}{{\bm{w }}_k}}$, and ${\bm{n}}_{\rm C}$ is additive white Gaussian noise with covariance $\sigma^2 \bm{I}_{L_{\rm C}}$,  independent of $\bm{x}$. 
	Previous works which characterized bounds on the achievable rates of index modulation schemes, e.g., \cite{Wen2016,Wen2017}, assumed that the channel input includes a digitally modulated symbol whose parameters are exploited to convey additional information via index modulation. Here, the primary user is the radar functionality, and the channel input ${\bm{x}}$ in \eqref{RxModel1} is a radar waveform. The information bits are embedded in $\bm x = \sum_{k = 0}^{K - 1} {{{\bm{P}}_k}{{\bm{w }}_k}}$, via the set of carrier frequencies, encapsulated in $\{ {{\bm{w }}_k} \}$, and their antenna allocation, modeled via $\{ {{\bm{P}}_k}\}$. The following achievable rate study is thus specifically tailored for the  statistical characterization of  ${\bm{x}}$ which arises in \ac{majorcom}.
	
	Based on the transmission scheme detailed in Section~\ref{sec:system}, we define a set $\mathcal{X} \subset \mathds{C}^{L_{\rm C}}$ that contains all the possible transmitted signal vectors $\bm x$, whose   carnality   is $|\mathcal{X}| =  \frac{{M!}}{{K!\left( {M - K} \right)!}}\frac{{L_R}!}{{\left( {{L_K}!} \right)}^K}$. 
	Assuming that the codewords are equally distributed, it holds that $\bm{x}$ is uniformly distributed over $\mathcal{X}$. 
	Consequently, the channel output ${{\bm{y}}_{\rm C}}$ obeys a Gaussian mixture (GM) distribution with equal priors. Let $\mvn{\ui}{\mean{}}{\var{}}{L_{\rm C}}$ denote the probability density function (PDF) of an $L_{\rm C}\times 1$ proper-complex Gaussian  vector with mean  $\mean{}\in\mathbb{C}^{L_{\rm C}}$ and covariance matrix $\var{}\in\mathbb{C}^{L_{\rm C}\times L_{\rm C}}$, where $\ui$ is the realization of the random vector. Then, the PDF of  ${\bm{y}}_{\rm C} $ is 
	\begin{equation}
	\label{eqn:GMpdf}
	f_{{\myVec{y}}_{\rm C}}\left( {\ui}\right) = \frac{1}{|\mySet{X}|}\sum_{\myVec{x}^{(i)} \in \mySet{X}} \mvn{\ui}{\myMat{H}\myVec{x}^{(i)} }{\sigma^2 \myMat{I}_{L_{\rm C}}}{L_{\rm C}}.
	\end{equation}
	 
	Using the input-output relationship of the channel, we can characterize the achievable rate. Let  $I(\cdot ; \cdot)$ and $h(\cdot)$ denote the mutual information and differential entropy, respectively. Since the channel in \eqref{RxModel1} is memoryless, its achievable rate is given by the single letter characterization \cite{Cover2006}
	\begin{eqnarray}
	{R_{\rm C}} &=& I\left( {{\myVec{x}};{{\myVec{y}}_{\rm C}}} \right) = h\left( {{{\myVec{y}}_{\rm C}}} \right) - h\left( {\left. {{{\myVec{y}}_{\rm C}}} \right|{\myVec{x}}} \right) \notag\\
	&{=}& h\left( {{{\myVec{y}}_{\rm C}}} \right) - h\left( {{{\myVec{n}}_{\rm C}}} \right) \label{eq:MutualInform1} \\ 
	&{=}& h\left( {{{\myVec{y}}_{\rm C}}} \right) - L_{\rm C} \cdot \log_2 \left(\pi \cdot e \cdot \sigma^2 \right) ,
	\label{eqn:AchievableRate}
	\end{eqnarray}
	where \eqref{eq:MutualInform1} holds since $\myVec{x}$ is independent of $\myVec{n}_{\rm C}$, and \eqref{eqn:AchievableRate} is the differential entropy of proper-complex Gaussian  vectors. 
	 
	In order to evaluate \eqref{eqn:AchievableRate}, one has to compute the differential entropy of the GM random vector ${\myVec{y}}_{\rm C}$. While there is no closed-form analytic expression for the differential entropy of GM random vectors \cite{Huber2008}, a lower bound on the achievable rate can be obtained, as stated in the following proposition:
	\begin{Prop}
		\label{pro:Rate1}
		The achievable rate of the proposed communication scheme is lower bounded by 
		\begin{equation*}
		{R_{\rm C}} \ge - \frac{1}{|\mySet{X}|}\sum_{\myVec{x}^{(i)} \in \mySet{X}} \log_2 f_{{\myVec{y}}_{\rm C}}\left( \myMat{H}{\myVec{x}^{(i)}}\right) - L_{\rm C} \cdot \log_2 \left(\pi \cdot e \cdot \sigma^2 \right),    
		\end{equation*}
		where $f_{{\myVec{y}}_{\rm C}}\left(  \cdot \right) $ is given in \eqref{eqn:GMpdf}.
	\end{Prop}
	\begin{proof}
		The proposition follows from lower bounding $h\left( {{{\myVec{y}}_c}} \right) $ using \cite[Thm. 2]{Huber2008}. 
	\end{proof}
	A trivial upper bound on ${R_{\rm C}}$ is obtained by noting that $\myVec{x}$ is uniformly distributed over the discrete set $\mySet{X}$, thus,
	\begin{equation}
	\label{eqn:Rate2}
	{R_{\rm C}} \le  h(\myVec{x}) = \log_2 |\mySet{X}|. 
	\end{equation} 
	This upper bound can be approached at sufficiently high SNRs where the codewords are reliably distinguishable. 
	We note that \eqref{eqn:Rate2} implies that the number of bits which can be conveyed in each pulse cannot be larger than the number of bits needed for representing the different codewords.  The upper bound in \eqref{eqn:Rate2} can be approximated using Stirling's formula via \eqref{eq:approximate_bits2}.
	 
	The achievable rate analysis provides a measure for quantifying the communication capabilities of \ac{majorcom}. In the numerical study in Section \ref{sec:sim} we demonstrate that in low SNRs, \ac{majorcom} is capable of achieving higher rates than using  individual dedicated communication waveforms, without interfering or even affecting the radar performance. Nonetheless, this information-theoretic framework does not account for practical considerations such as computational burden at the receiver, motivating the reduced complexity implementation presented in the following section.
	
	\vspace{-0.2cm}
	\section{Reduced Decoding Complexity Implementation}
	\label{sec:complexity}
		\vspace{-0.1cm}
	As discussed in the introduction, one of the major benefits of \ac{majorcom} stems from its usage of narrowband signals and relatively low computational complexity, which imply that it can be implemented using simple hardware components. However, while generating and transmitting the communication signal by \ac{majorcom} does not require heavy computations, decoding the transmitted index-modulated message by the communication receiver may entail a substantial computational burden.   
	Consequently, in this section we propose methods for reducing the decoding complexity. 
	
	We begin by discussing the optimal maximum likelihood (ML)  symbol decoding scheme in Subsection~\ref{subsec:ML}. Then, we present two approaches for mitigating its complexity: In  \ref{subsec:detection} we propose a sub-optimal decoding method, which affects only the communication receiver. Then, we propose a modified codebook design which facilitates decoding by reducing the number of codewords used by \ac{majorcom} in Subsection~ \ref{subsec:Codebook}. 
	The change of codebook may affect the radar beam pattern, however the simulation results present later in Section \ref{sec:sim} demonstrate that this change has minimum influence on range, Doppler and angle estimates of radar targets. 
	Those two approaches are independent of each other, and can be used either simultaneously or individually, depending on the computational abilities of the communications receiver.

	\vspace{-0.2cm}
	\subsection{Optimal ML Decoder}
	\label{subsec:ML}
	\vspace{-0.1cm}	
	To detect the conveyed symbols, the receiver estimates both the selected frequencies and allocated antenna indices. 	
	Since the entries of the noise matrix $\bm N_{\rm C}$ are i.i.d. Gaussian and the codewords are equiprobable, the detector which minimizes the probability of error is the ML estimator of the frequency indices $\{c_k\}$ and the antenna allocations $\{\bm P_k\}$ \cite[Ch. 5.1]{Goldsmith2005}. From \eqref{eq:YC}, the ML estimator is given by
	\begin{equation}
	\label{eq:ML} 
	\left\{\hat{c}_k,\widehat{\bm P}_{k}\right\}_{k=0}^{K-1}  
	= \mathop{\arg \min} \limits_{\{c_k, \bm P_k\}} \left\| \bm Y_{\rm C}- \sum \limits_{k=0}^{K-1} \bm H \bm P_k \bm w_k {\bm \psi}_{c_{k}}^T\right\|_F^2, 
	\end{equation}
	where $\left\| \cdot \right\|_F$ denotes the Frobenius norm. Since the frequency indices $\{c_k \}$ and the selection matrices $\{\bm P_k \}$ are integers and binary matrices, respectively, the above problem is generally NP-hard. In particular, solving \eqref{eq:ML} involves exhaustively searching over $\mho$ and $\mathcal{P}$,  resulting in high computational complexity.
	 This increased complexity settles with the fact that optimal index modulation decoding is typically computationally complex \cite{Basar2016}. 
	 
	 Various  low complexity methods have been  proposed for different forms of index modulation, see, e.g., \cite[Tbl. 1]{Basar2016}. However, as our form of index modulation, in which all the transmitted information is embedded in the selection of the frequencies and their allocation among antennas (without additional digital modulation signals), is unique, in the next subsection we design a dedicated low-complexity decoder.  

	\vspace{-0.2cm}
	\subsection{Low Complexity Receiver Design}
	\label{subsec:detection}
	\vspace{-0.1cm}	
	Here, we present a sub-optimal detection method. Instead of jointly estimating $\{c_k, \bm P_k\}$ in \eqref{eq:ML}, our proposed strategy operates in an iterative manner: It first initializes the frequency estimates $\{c_k\}$ using sparse recovery, possibly via simple fast Fourier transformation (FFT) followed by thresholding. Then, we iteratively recover the spatial selection matrices $\{\bm P_k\}$, and refine the estimation of $\{c_k\}$ in an alternating fashion. 
	\subsubsection{Frequency Initialization}
	\label{subsubsec:ML_frequency}
	In the first step, we obtain an initial estimation of the transmitted frequencies. To that aim, we rewrite the model \eqref{eq:YC} as
	\begin{equation}
	\label{eq:fft}
	\bm Y_{\rm C}^T = \bm \Psi \bm A + \bm N_{\rm C}^T,
	\end{equation}
	where $\bm \Psi =\left[\bm \psi_0, \bm \psi_1,\dots,\bm \psi_{M-1}\right] \in \mathbb{C}^{L_T \times M}$ contains all $M$ sub-bands, and thus is a-priori known.   
	The matrix $\bm A \in  \mathbb{C}^{M \times L_{\rm C}}$ depends on the frequency indices $\{c_k \}$: When there exists an index $c_k =m$,  the transpose of the $m$-th row of $\bm A$ is given by
	\begin{equation}
	\label{eq:A}
	\left[ \bm A^T \right]_{m} =  \bm H \bm P_k \bm w_k \in \mathbb{C}^{L_{\rm C}},
	\end{equation}
	while otherwise $\left[ \bm A^T \right]_{m} = \bm 0_{L_{\rm C}}$. We regard $\bm A$ as an unknown variable, which has to be estimated. 
	After $\bm A$ is estimated as $\widehat{\bm A}$, the frequency indices $\{c_k\}$ are recovered from the non-zero rows (or the $K$ rows with largest norms) of $\widehat{\bm A}$.
	
	When the number of active frequencies is sufficiently smaller than the number of available frequencies, i.e., $K \ll M$, \eqref{eq:fft} becomes a typical sparse recovery problem, and  $\widehat{\bm A}$ be can be obtained using any sparse recovery method  \cite{Eldar2012}. 
	We note that when the pulse duration is  an integer multiple of $1/\Delta f$, i.e., $T_p = n/\Delta f$, $  n \in \mathbb{Z}^+$, then $L_T = nM$ and $\bm \Psi$ in \eqref{eq:fft} consists of $M$ columns from the $L_T \times L_T$ FFT matrix. In such cases, in which the columns of $\bm \Psi$ are orthogonal (or approximately orthogonal), it is noted that simple projection and thresholding may achieve comparable support recovery performance as more computationally complex iterative sparse recovery methods. In particular, when  the columns of $\bm \Psi$ are orthogonal, projection and thresholding recovers $\widehat{\bm A}$ via   
	\begin{equation}
	\label{eq:AhatFFT}
	    \widehat{\bm A} =  \bm \Psi^H \bm Y_{\rm C}^T,
	\end{equation} 
	which can be computed using FFT. We then sort the norms of rows, $\left\| \left[ \bm A^T \right]_{m} \right\|_2$, in a descending order, and identify the first $K$ rows, 
	which correspond to the frequency indices $\left\{c_k\right\}$. 
	
	The aforementioned simplified scheme is most suitable when $T_p \approx n/\Delta f$, and its main benefit is its low computational complexity. When this approximation does not hold, one can utilize any sparse recovery method for obtaining $\widehat{\bm A}$. 
	\subsubsection{Spatial Decoder}
	\label{subsubsec:ML_spatial}
	After the frequency indices are recovered as $\{\hat{c}_k\}$, the ML estimator (\ref{eq:ML}) becomes
	\begin{equation}
	\label{eq:ML_spatial}
	\left\{\widehat{\bm P}_{k}\right\}_{k=0}^{K-1}  
	= \mathop{\arg \min} \limits_{\{\bm P_k\}} \Big\| \bm Y_{\rm C}- \sum \limits_{k=0}^{K-1} \bm H \bm P_k \bm w_k {\bm \psi}_{\hat{c}_{k}}^T\Big\|_F^2,
	\end{equation}
	which jointly optimizes $K$ selection matrices $\left\{ \bm P_k \right\}_{k=0}^{K-1}$ and can be solved by exhaustive search over $\mathcal{P}$. As directly solving \eqref{eq:ML_spatial} may still be difficult, we next introduce a greedy approach that solves each selection matrix $\bm P_k$ sequentially to reduce the computational burden.
	 
	Denote by $\hat{c}_0,\hat{c}_1,\dots,\hat{c}_{K-1}$ the obtained frequency indices $\left\{ c_k\right\}$ in such an order that the corresponding rows of $\widehat{\bm A}$ satisfy $\left\| \left[ \bm A^T \right]_{\hat{c}_0} \right\|_2 \ge \left\| \left[ \bm A^T \right]_{\hat{c}_1} \right\|_2 \ge \dots \ge \left\| \left[ \bm A^T
	\right]_{\hat{c}_{K-1}} \right\|_2$. 
	According to  (\ref{eq:A}), we write the $\hat{c}_k$-th row of $\bm A$ as
	\begin{equation}
	\label{eq:AHp}
	\left[\widehat{\bm A}^T \right]_{\hat{c}_k}=\widetilde{ \bm H }\bm p_k + \bm n_{k},
	\end{equation}
	where $\widetilde{ \bm H }: = { \bm H }{\rm diag}\left( \bm w_k  \right)$; ${\rm diag}\left( \bm w_k  \right)$ denotes the diagonal matrix with entries defined in $\bm w_k $; $\bm p_k \in \{0,1\}^{L_{\rm R}}$ contains the diagonal entries in $\bm P_k$; and $\bm n_{k}$ denotes the estimate errors in $ \big[ \widehat{\bm A}^T \big]_{c_k}$.  Recall that in each pulse, every antenna is assigned to a single frequency, and thus $\sum_{i=0}^{K-1}{\bm p_i} = \bm 1_{L_{\rm R}}$, implying that $\bm p_k \land \left( \sum_{i=0}^{k-1}{\bm p_i} \right)= \bm 0_{L_{\rm R}}$, where $\land$ denotes entry-wise logical and operation. The fact that the unknown vectors $\{p_k\}$ take binary values and are subject to this joint constraint implies that \eqref{eq:AHp} should not be treated as $K$ individual linear recovery problems, giving rise to the following sequential approach. Here, we assume that $\bm p_0, \bm p_1,\dots, \bm p_{k-1}$ have been recovered prior to $\bm p_k$.
	Then, we use \eqref{eq:AHp} to formulate the recovery of $\bm p_k$ as: 
	\begin{equation}
	\label{eq:leastsquare}
	\begin{split}
	\hat{\bm p}_k = &\mathop{\arg \min}\limits_{\bm p_k} \Big\| \left[\widehat{\bm A}^T \right]_{c_k}  - \widetilde{ \bm H }\bm p_k \Big\|_2^2, \\
	{\rm s.t.\ }& \bm p_k \land \left( \sum_{i=0}^{k-1}{\bm p_i} \right)= \bm 0_{L_{\rm R}},\  \left\|\bm p_k\right\|_1 = L_K.
	\end{split}
	\end{equation}
	Note that \eqref{eq:leastsquare} should be solved using exhaustive search due to its non-conventional constraints and since ${\bm p_k}$ takes binary values.  
	There are $L_{\rm R}-kL_{K} \choose L_{K}$ possible values for each $\bm p_k$, and at most a total of $K{L_{\rm R} \choose L_{K}}$ evaluations should be carried out to recover all vectors $\{\bm p_k\}$. Compared with the optimal ML method  (\ref{eq:ML_spatial}), which requires approximately $\left| \mathcal{P} \right| \approx {L_{\rm R} \choose L_{K}}^K $ searches once the frequency indices $\{c_k\}$ are recovered, the sub-optimal method reduces the complexity significantly. 
	In our numerical analysis in Section \ref{sec:sim} we show that the proposed low-complexity decoder is capable of achieving BER performance which is comparable with the computationally complex ML decoder.
	
    The proposed method obtains a coarse estimate of $\{c_k,\bm P_k\}$, and the corresponding decoder is summarized in Alg. \ref{alg:noniter}. 
	This coarse estimate can be later refined by updating $\{c_k\}$ and $\{\bm P_k\}$, as well as $\widehat{\bm A}$ which is used in estimating both $\{c_k\}$ and $\{\bm P_k\}$, in an alternating manner. The method to update $\{\bm P_k\}$ using an estimate of $\{c_k\}$ and $\widehat{\bm A}$ is based on the above, while the refining of $\{c_k\}$ and $\widehat{\bm A}$ based on  an estimate of $\{\bm P_k\}$  is detailed in the subsequent frequency refinement step.
	
	\begin{algorithm}
	\caption{Non-iterative low complexity decoder}
	\label{alg:noniter}
	{\bf Input:}\ $\bm Y_C$, $\bm \Psi$, $K$.\\
	{\bf Steps:}
    	\begin{algorithmic}
    	    \State (1) Compute $\widehat{\bm A}$ via sparse recovery, possibly using \eqref{eq:AhatFFT} if $T_p \approx n/\Delta f$, and calculate the norms of the rows of  $\widehat{\bm A}$.
    	    \State (2) Sort these norms in a descend order, recovering $\{\hat{c}_k\}$.
    	    \State (3) Apply ML-based spatial decoder via \eqref{eq:ML_spatial}, or perform greedy spatial decoding, i.e., sequentially solve \eqref{eq:leastsquare}.
    	\end{algorithmic}
	{\bf Output:}\ $\{\hat{c}_k,\widehat{\bm P}_k\}_{k = 0}^{K-1}$.
	\end{algorithm}
	
	%
	\subsubsection{Frequency Refinement} 
	With the estimates $\{ \widehat{\bm P}_k \}_{k=0}^{K-1}$, we then refine the frequency codes $\{c_k\}_{k=0}^{K-1}$. According to \eqref{eq:ML}, the optimization problem becomes
	\begin{equation}
	\label{eq:ML_frequency}
	\{ \hat{c}_k\} = \mathop{\arg \min} \limits_{\{ c_k\}} \Big\| \bm Y_{\rm C}- \sum \limits_{k=0}^{K-1}  \widetilde{\bm w}_k {\bm \psi}_{c_{k}}^T\Big\|_F^2,
	\end{equation}
	where $\widetilde{\bm w}_k:= \bm H \widehat{\bm P}_k \bm w_k \in \mathbb{C}^{L_{\rm C}}$. Since the set $\{ c_k\}$ consists of $K$ distinct indices in the range $\{0,1,\ldots,M-1\}$, \eqref{eq:ML_frequency} should be solved via  exhaustive search, requiring a total of ${M \choose K}$ evaluations. Similarly to the aforementioned spatial decoder that utilizes greedy approach, one may also estimate the frequency codes sequentially to reduce the computation, as detailed next.
	
	Particularly, when recovering $c_{k}$, let $\{\hat{c}_{m} \}_{m=0}^{k-1} $ be the previously obtained frequency indices. The estimation of  $c_{k}$ can be formulated by rewriting \eqref{eq:ML_frequency} as
	\begin{equation}
	\label{eq:greedy_frequency}
	\begin{split}
	\hat{ c}_k = \mathop{\arg \min} \limits_{ c_{k}} \Big\| &\bm Y_{\rm C}- \sum \limits_{m=0}^{k-1}  \widetilde{\bm w}_{m} {\bm \psi}_{\hat{c}_{m}}^T
	- \bm H \widehat{\bm P}_{k} {\bm w}_{k}{\bm \psi}_{{c}_{k}}^T
	\Big\|_F^2,\\
	{\rm s.t. }\ & c_{k} \in \{ 0,1,\dots, M-1 \}\backslash \{\hat{c}_{m} \}_{m=0}^{k-1}.
	\end{split}
	\end{equation}	
	In this greedy approach, only a total of $\sum_{k=0}^{K-1}(M-k) = KM - \frac{K(K-1)}{2}$ evaluations are required, which is much less than its exhaustive search counterpart. 
	
	Next, we use the estimates of $\{ \hat{c}_k, \widehat{\bm P}_k \}$ to obtain a refined estimate of $\widehat{\bm A}$, which is used  in the greedy spacial decoder \eqref{eq:leastsquare}. In the initial steps (as indicated in Alg. \ref{alg:noniter}), in which an estimate of $\{ \hat{c}_k, \widehat{\bm P}_k \}$ is not available, $\widehat{\bm A}$ is computed with sparse recovery. Having obtained $\{ \hat{c}_k, \widehat{\bm P}_k \}$, we can refine the value of $\widehat{\bm A}$. In particular, by \eqref{eq:A},  $\widehat{\bm A}$ can now be computed by setting its $\hat{c}_k$-th row to
	\begin{equation}
	\label{eq:Ahat}
	\left[ \widehat{\bm A}^T \right]_{\hat{c}_k} =  \bm H \widehat{\bm P}_k \bm w_k,
	\end{equation}
	while fixing the remaining rows to be the all-zero vector. The resulting algorithm, which uses the Spatial Decoder and Frequency Refinement steps to update $\{ {\bm P}_k \}_{k=0}^{K-1}$ and $\{c_k\}_{k=0}^{K-1}$, respectively, is summarized in Alg. \ref{alg:iter}.
	
	
	\begin{algorithm}
	\caption{Iterative low complexity decoder}
	\label{alg:iter}
	{\bf Input:}\ $\bm Y_C$, $\bm \Psi$, $K$, maximal iteration $i_{\max}$.\\
	{\bf Initialization:}
    	\begin{algorithmic}
    	    \State (1) $i \leftarrow 1$.
    	    \State (2) Obtain  $\{\hat{c}_k^{(0)},\widehat{\bm P}_k^{(0)}\}_{k = 0}^{K-1}$ and $\widehat{\bm A}^{(0)}$  using Alg. \ref{alg:noniter}.
        \end{algorithmic}
	{\bf While} $i < i_{\max}$:
    	\begin{algorithmic} 
    	    \State  (3) Obtain $\{c_k^{(i)}\}_{k = 0}^{K-1}$ via ML-based frequency refinement  \eqref{eq:ML_frequency}, or by sequentially solving \eqref{eq:greedy_frequency}.
    	    \State (4) Evaluate $\widehat{\bm A}^{(i)}$ using \eqref{eq:Ahat}.
    	    \State (5) Compute $\{\widehat{\bm P}_k^{(i)}\}_{k = 0}^{K-1}$  by applying ML-based spatial decoding \eqref{eq:ML_spatial}, or by sequentially solving \eqref{eq:leastsquare}.
    	     \State  (6) $i \leftarrow i+1$.
    	\end{algorithmic}
	{\bf Output:}\ $\{\hat{c}_k^{(i_{\max}-1)},\widehat{\bm P}_k^{(i_{\max}-1)}\}_{k = 0}^{K-1}$.
	\end{algorithm}

	\subsection{Codebook Design}
	\label{subsec:Codebook}
	\vspace{-0.1cm}
	In the description of \ac{majorcom} in Section \ref{sec:system}, all possible options in $\mho$ and $\mathcal{P}$ are coded uniquely and used to carry different symbols. 
	Fully exploiting the variety of these sets allows the achievable rate to approach the upper bound in \eqref{eqn:Rate2} at sufficiently high SNR, as the different codewords can be reliably distinguished from one another. However, the computational complexity required to properly decode the message grows rapidly with the cardinality of these sets. Particularly, while detecting the used frequencies from $\mho$ can be implemented in a low complexity manner at the cost of some performance reduction, recovering the antenna allocation from $\mathcal{P}$ typically requires an exhaustive search,
	as discussed in the previous subsection. Therefore, in order to facilitate accurate decoding under computational complexity constraints, we now propose a codebook design which makes full use of $\mho$ while utilizing a subset of $N_b$ codewords from $\mathcal{P}$, thus balancing achievable rate and computational burden at the receiver.  
	  
	  
	Our goal is to design a constellation set, which is a subset of $\mathcal{P}$, such that the ability of the receiver to distinguish between different codewords is imporved. To that aim, we first discuss the design criterion, which yields a high dimensional NP-hard max-min problem. To solve it, we first apply a dimension reduction approach, after which we propose a sub-optimal solution. 
	 
	\subsubsection{Design Criterion}
	\label{subsubsec:DesignCriterion}
	When the impact on radar function is not accounted for, the proper codebook design objective is to maximize the minimum distance between any two codewords, $\{\bm P_k^{(i)}\}_{k=0}^{K-1} $ and $\{\bm P_k^{(j)}\}_{k=0}^{K-1} $, or equivalently,  $\{\bm p_k^{(i)}\}_{k=0}^{K-1} $ and $\{\bm p_k^{(j)}\}_{k=0}^{K-1} $. In particular, it follows from (\ref{eq:ML_spatial}) that the distance,
	\begin{equation}
	\bigg\| \sum \limits_{l=0}^{K-1} \bm H \bm P_l^{(i)} \bm w_l {\bm \psi}_{c_{l}^{(i)}}^T- \sum \limits_{k=0}^{K-1} \bm H \bm P_k^{(j)} \bm w_k {\bm \psi}_{c_{k}^{(j)}}^T\bigg\|_F^2,
	\end{equation}
	dominates the error probability between the $i$-th and $j$-th symbols. Since we optimize the minimum distance with respect to the antenna allocations $\{\bm p_k\}$, we henceforth focus on the setting where the set of frequency indices are the same in those symbols, i.e., $\big\{ c_k^{(i)} \big\}_{k=0}^{K-1}$ equals $\big\{ c_k^{(j)} \big\}_{k=0}^{K-1}$. This setting generally 
	leads to a smaller distance in comparison with the unequal case, and can thus be considered as a worst case scenario. 
	
	When the frequency modulations are orthogonal, i.e., ${\bm \psi}_{m_1}^H{\bm \psi}_{m_2} = 0$, $m_1\neq m_2$, $m_1,m_2 \in \mySet{M}$, which holds when $T_p$ is an integer multiple of $1/\Delta f$, the distance between two codewords can be simplified to
	\begin{equation}
	\label{eq:dist_H}
	{\text{H-Dist}}_{i,j} := \sum \limits_{k=0}^{K-1}\left\| \widetilde{ \bm H }\bm p_k^{(i)} - \widetilde{ \bm H }\bm p_k^{(j)} \right\|_2^2.
	\end{equation} 
	The distance \eqref{eq:dist_H} is upper bounded by the largest eigenvalue of $ \widetilde{ \bm H }$ times $\text{Dist}_{i,j}$, which is defined as
	\begin{equation}
	\label{eq:dist_vec}
	{\text{Dist}}_{i,j} := \sum \limits_{k=0}^{K-1}\left\| \bm p_k^{(i)} - \bm p_k^{(j)} \right\|_2^2.
	\end{equation} 
	We propose a codebook design to find a subset $\mathcal{P}^{\#} \subset \mathcal{P}$ of cardinality $N_b$ that maximizes the distance
	\begin{equation}
	\label{eq:maxmindistance}
	\max_{\mathcal{P}^{\#} \subset \mathcal{P}} \min_{i,j \in \mathcal{P}^{\#}, i\neq j} {\text{Dist}}_{i,j},\quad {\rm s.t.}\ \left| \mathcal{P}^{\#} \right| = N_b.
	\end{equation}
	We note that \eqref{eq:maxmindistance} is still NP-hard to solve. Although the objective in \eqref{eq:maxmindistance} can be considered as the minimal Hamming distance, standard codebook desgins based on this criterion, see \cite[Ch. 8]{Goldsmith2005}, cannot be used here. The reason is that our codewords are subject to the additional unique constraint $\sum_{k=0}^{K-1}\bm p_k^{(i)}  = \bm 1_{L_{\rm R}}$, which does not appear in standard binary codebooks. Thus, we propose a codebook design based on projection into a lower dimensional plane, described next.
	 
	\subsubsection{Dimension Reduction of the Constellation Set}
	\label{subsubsec:DimensionReduction}
	We propose to project the original codewords into a real-valued $L_{\rm D}$-dimensional plane, i.e., $ \{\bm p_k^{(i)} \} \mapsto \widetilde{\bm p}^{(i)} \in \mathbb{R}^{L_{\rm D}}$, such that the distances between codewords are maintained
	\begin{equation}
	\label{eq:distance_transform}
	\widetilde{d}\left<i,j\right> := \left\| \widetilde{\bm p}^{(i)} - \widetilde{\bm p}^{(j)} \right\|_2^2 = {\text{ Dist}}_{i,j}, 
	\end{equation}
	where $ i,j = 0,1,\dots,\left|\mathcal{P} \right|-1$. It is easy to verify that \eqref{eq:distance_transform} holds when there exist an orthogonal matrix $\bm U \in \mathbb{R}^{KL_{\rm R} \times KL_{\rm R}}$ and a constant vector $\bm a \in \mathbb{R}^{KL_{\rm R} }$ such that
	\begin{equation}
	    \label{eq:pp}
	    {\bm p}^{(i)} = \bm U \left[ \begin{array}{c}
	\widetilde{\bm p}^{(i)}\\
	\bm 0_{KL_{\rm R} - L_{\rm D}}\\
\end{array} \right] + \bm a, 
	\end{equation}
	where $\bm p^{(i)} := \left[ \bm p_0^T, \bm p_1^T, \dots, \bm p_{K-1}^T\right]^T \in \{0,1\}^{KL_{\rm R}}$.
	
	To find such $\bm U$, $\bm a$ and $\widetilde{\bm p}^{(i)}$, we use  principal component analysis (PCA) \cite{Maaten2009}. Denote the codebook matrix by $\bm D := \left[\bm p^{(0)}, \bm p^{(1)}, \dots, \bm p^{(|\mathcal{P}|-1)} \right] \in \{0,1\}^{KL_{\rm R} \times |\mathcal{P}|}$, and the dimension reduced matrix by $ \widetilde{\bm D} := \left[\widetilde{\bm p}_0, \widetilde{\bm p}_1, \dots, \widetilde{\bm p}_{|\mathcal{P}|-1} \right]\in \mathbb{R}^{L_{\rm D} \times |\mathcal{P}|}$, respectively. Then, \eqref{eq:pp} becomes
	\begin{equation}
	    \label{eq:DD}
	    {\bm D} = \bm U \left[ \begin{array}{c}
	\widetilde{\bm D}\\
	\bm 0_{(KL_{\rm R} - L_{\rm D})\times |\mathcal{P}|}\\
\end{array} \right] + \bm a \cdot \bm 1_{|\mathcal{P}|}^T. 	\end{equation}

	Noticing that $\bm p^{(i)}$ has identical average, i.e., $\frac{1}{KL_{\rm R}} \bm 1_{KL_{\rm R}}^T \bm p ^{(i)} =  \frac{1}{K}$,   
	we first normalize columns of $\bm D$ to zero mean by
	\begin{equation}
	    \overline{\bm D} = \bm D - \frac{1}{K} \bm 1_{KL_{\rm R} \times |\mathcal{P}|}.
	\end{equation}
	
	With $\overline{\bm D}$,  we then perform SVD decomposition on $\overline{\bm D}$, i.e., 
	\begin{equation}
	    \overline{\bm D} = \bm U \bm \Sigma \bm V^T,
	\end{equation}
	where $\bm U \in \mathbb{R}^{KL_{\rm R} \times KL_{\rm R}}$ and $\bm V \in \mathbb{R}^{|\mySet{P}| \times |\mySet{P}|}$ are unitary matrices, $\bm U \bm U^T = \bm U^T \bm U = \bm I_{KL_{\rm R}}$, $\bm V \bm V^T = \bm V^T \bm V = \bm I_{|\mySet{P}|}$, and $\bm \Sigma \in \mathbb{R}^{ KL_{\rm R} \times |\mySet{P}|}$ is a diagonal matrix with $[\bm \Sigma]_{i,i}$, $i \le KL_{\rm R}$, being the singular values of $\overline{\bm D}$. 
	We estimate $L_{\rm D}$, which is often regarded as the intrinsic dimension of the original codewords, as the number of nonzero singular values, i.e., the rank of $\overline{\bm D}$, and the transpose of the new codewords are given by
	\begin{equation}
	    \widetilde{\bm D}^T = \bm V \left[  \bm \Sigma ^T\right]_{\{0,1,\dots, L_{\rm D}-1\}} .
	\end{equation}
	Let $\bm a = 1/K \bm 1_{KL_{\rm R}}$ and it can be verified that \eqref{eq:DD} holds and codewords $\widetilde{\bm D}$ preserve the distances as stated in \eqref{eq:distance_transform}. 
	

	It is worth noting that the special structure of $\bm p^{(i)}$ results in some symmetry of the distances ${\text{Dist}}_{i,j}$. To see this, we define the distance matrix $\bm R \in \mathbb{Z}^{\left|\mathcal{P} \right| \times \left|\mathcal{P} \right|}$ with entries
	\begin{equation}
	\label{eq:Rmatrix}
	\left[{\bm R}\right]_{i,j} = {\text{Dist}}_{i,j}, \quad  i,j = 0,1,\dots,\left|\mathcal{P} \right|-1.
	\end{equation}  
	The distance matrix has the following properties.
	\begin{Prop}
		\label{prop:R_diag}
		The matrix $\bm R$ is symmetric and its diagonal entries are zeros, i.e,. $\left[{\bm R}\right]_{i,j} =  \left[{\bm R}\right]_{j,i}$ and $\left[{\bm R}\right]_{i,i} = 0$,  $i,j = 0,1,\dots,\left|\mathcal{P} \right|-1$. Furthermore, each row of $\bm R$ is a permutation of the first row in $\bm R$.
	\end{Prop}
	\begin{proof}
	A proof is given in the Appendix.
	\end{proof}

	Proposition \ref{prop:R_diag} implies that, given a set $\mathcal{P}$ of different possible antenna allocation codewords,  the calculation of 
	the distance matrix $\bm R$ is far less computationally complex than evaluating the distance between each possible pair of elements of $\mathcal{P}$ in a straightforward manner.

	We take $L_{\rm R} = 4$, $K=2$ and $L_K=2$ as an example to demonstrate the dimension reduction. When $L_{\rm D}= 2, 3$, we find that the projected codebook can be visualized conveniently using ${\widetilde{\bm p}}$. To see this, recall that there are 6 possible spatial selection patterns as explained by (\ref{eq:P_example}). The original codewords, $\left\{\bm p_0^{(i)},\dots, \bm p_{K-1}^{(i)} \right\}$, $i=0,\dots,5$, have $KL_{\rm R} =8$ dimensions, and are difficult to display. The entries of the distance matrix here are given by
	\begin{equation*}
	   [\bm R]_{i,j} = \begin{cases}
	   0 & i = j, \\
	   8 & |i-j| = 3, \\
	   4 & {\rm otherwise}.
	    \end{cases}
	\end{equation*} 
	After dimensionality reduction, one obtains the following three-dimensional representation of the codewords: $\widetilde{\bm p}^{(0)}=[0,\sqrt{2},0]^T$, $\widetilde{\bm p}^{(1)}=[\sqrt{2},0,0]^T$, $\widetilde{\bm p}^{(2)}=[0,0,\sqrt{2}]^T$, and $\widetilde{\bm p}^{(3)}=-\widetilde{\bm p}^{(0)}$, $\widetilde{\bm p}^{(4)}=-\widetilde{\bm p}^{(1)}$, $\widetilde{\bm p}^{(5)}=-\widetilde{\bm p}^{(2)}$. We can verify that $\widetilde{d}\left<i,j\right> = \left[{\bm R}\right]_{i,j}$.
	The resulting three-dimennsional constellation set is depicted in Fig. \ref{fig:constellation_3D}.
	\begin{figure}[!h]
	\centering
	\includegraphics[width=2.5in]{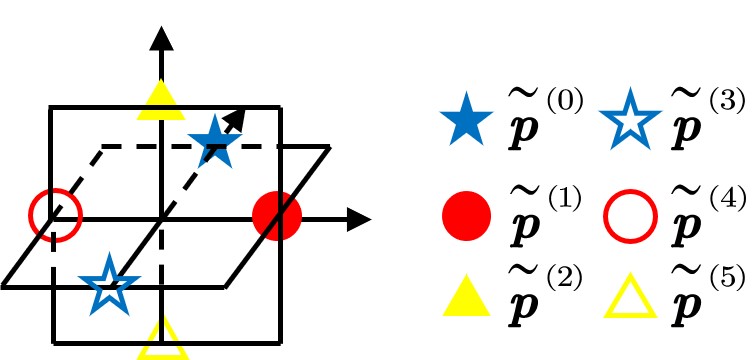}
	\caption{A constellation of dimension reduced codewords, $L_{\rm R} = 4$, $K=2$ and $L_K=2$. }
	\label{fig:constellation_3D}
	\end{figure}
	\subsubsection{Design of the Constellation Set}
	\label{subsubsec:DesignConstellation}
	After dimension reduction, the codebook design problem (\ref{eq:maxmindistance}) becomes
	\begin{equation}
	\label{eq:maxmindistance2D}
	\max_{\mathcal{P}^{\#}\subset \mathcal{P}} \min_{i,j \in \mathcal{P}^{\#}, i\neq j} \widetilde{d}\left<i,j\right> ,  \ {\rm s.t.}\ \left| \mathcal{P}^{\#} \right| = N_b.
	\end{equation}
	We propose the following sub-optimal approach to design a codebook based on \eqref{eq:maxmindistance2D}: Using clustering methods such as k-means, the codewords $\widetilde{\bm p}$ can be divided into $N_b$ classes. The codeword which is the nearest to the center point of the class is used to represent the class in the final codebook.
	Since clustering methods typically maximize the distances between classes, the proposed codebook is expected to have a large minimal distance, thus approaching the solution to \eqref{eq:maxmindistance2D}.
	
	Reducing the number of different antenna allocations affects the spatial agility and radiation pattern of the radar scheme, and thus potentially impacts the accuracy of range, Doppler or angular parameters of \ac{caesar}. Nonetheless, in the simulations study presented in Section \ref{sec:sim} it is numerically demonstrated that the radar performance degradation due to using the proposed reduced cardinality codebook is minimal. 

	\vspace{-0.2cm}
	\section{Simulations}
	\label{sec:sim}
	\vspace{-0.1cm}
	In this section we numerically evaluate the performance of \ac{majorcom}. Since the radar functionality of \ac{majorcom} is based on \ac{caesar} and is not affected by the communication subsystem, we focus here on the communication functionality of \ac{majorcom}, and refer to \cite{Huang2019b} for a detailed study of its radar performance.

	In particular, three aspects of the communication scheme are evaluated: First, in Subsection~\ref{subsec:RateSim}  the fundamental limits of the proposed system are compared to using different waveforms for communications and radar. Then,  the proposed low complexity decoders are numerically compared to the optimal ML decoder in Subsection~\ref{subsec:sim_decoder}. Finally, in Subsection~\ref{subsec:sim_code} the proposed reduced complexity codebook design approaches are evaluated along with their effect on radar performance. Throughout this study, the initial frequency is $f_c = 1.9$ GHz, the frequency spacing is $\Delta f = 10$ MHz, and the number of frequencies utilized at each pulse is $K=2$.
	
	
	\vspace{-0.2cm}
	\subsection{Achievable Rate}
	\label{subsec:RateSim}
	\vspace{-0.1cm} 
	Our achievable rate analysis quantifies the communication capabilities of \ac{majorcom}, facilitating its comparison to other configurations. As a numerical example, we consider a scenario with $4$ transmit and receive antennas, i.e., $L_{\rm R} = L_{\rm C} = 4$. The parameters of the proposed system are set to  $\theta = \frac{\pi}{4}$, $d = 10\frac{c}{f_c}$, and the number of available frequencies is $M=10$. The selection matrices used are given in \eqref{eq:P_example}. The overall number of codewords here is $|\mySet{X}| = 270$, i.e., the maximal number of bits that can be conveyed in each pulse is $\log_2 |\mySet{X}| \approx 8.1$. 
	We consider two settings for the channel matrix $\myMat{H}$: A spatial exponential decay channel, for which $\left[\myMat{H} \right]_{l_1,l_2} = e^{-\frac{1}{4}\left(|l_1-l_2|+j(l_1-l_2)\pi\right)}$; and Rayleigh fading, where the entries of $\myMat{H}$ are randomized from an i.i.d. zero-mean unit-variance proper-complex Gaussian distribution, and the achievable rate is averaged over $100$ realizations. 
	
	For each channel,
	we evaluate the lower and upper bounds on the achievable rate computed via Proposition \ref{pro:Rate1} and \eqref{eqn:Rate2}, respectively versus SNR, defined here as $1/\sigma^2$. This bound is compared to the rate achievable (in bits per channel use) when, instead of using the randomness of the radar scheme to convey bits, either the first antenna or the first two antennas are dedicated only for communications subject to a unit average power constraint, i.e., the same power as that of the radar pulse, neglecting the cross interference induced by radar and communications coexistence. {This study allows to understand when the achievable rate of \ac{majorcom}, which originates from radar transmission, is comparable to using ideal dedicated communication transmitters, which are  costly and induce mutual interference between radar and communications.}   The numerically evaluated achievable rates for the spatial decay channel and the Rayleigh fading channel are depicted in Figs. \ref{fig:MultiBand1}-\ref{fig:MultiBand2}, respectively.

	Observing Figs. \ref{fig:MultiBand1}-\ref{fig:MultiBand2}, we note that in relatively low SNRs, our proposed scheme achieves higher rates compared to using a dedicated communications antenna  element {\em without impairing the radar performance}. For Rayleigh fading channels, it is demonstrated in Fig. \ref{fig:MultiBand2} that  \ac{majorcom} is capable of outperforming a system with two dedicated communication antennas for SNRs not larger than $5$ dB. 
	As the SNR increases, using dedicated communication antennas outperforms our proposed system as more and more bits can be reliably conveyed in a single channel symbol. However, it should be emphasized that by allocating some of the antenna elements for communications, the radar performance, which is considered as the primary user in our case, is degraded. Furthermore, in order to avoid coexistence issues, which we did not consider here, the communications and radar signals should be orthogonal, e.g., use distinct bands, thus reducing the radar bandwidth. Finally, the computation of the achievable rate with dedicated antennas assumes the transmitter has \ac{csi} and does not account for the need to utilize constant modulus waveforms; it is in fact achievable using Gaussian signaling \cite[Ch. 9]{Cover2006}. Consequently, the fact that, in addition to the practical benefits of our proposed scheme and its natural coexistence with the radar transmission, it is also capable of achieving communication rates comparable to using dedicated communication antennas, illustrates the gains of \ac{majorcom}.
	\begin{figure}
		\centering
		\includefig{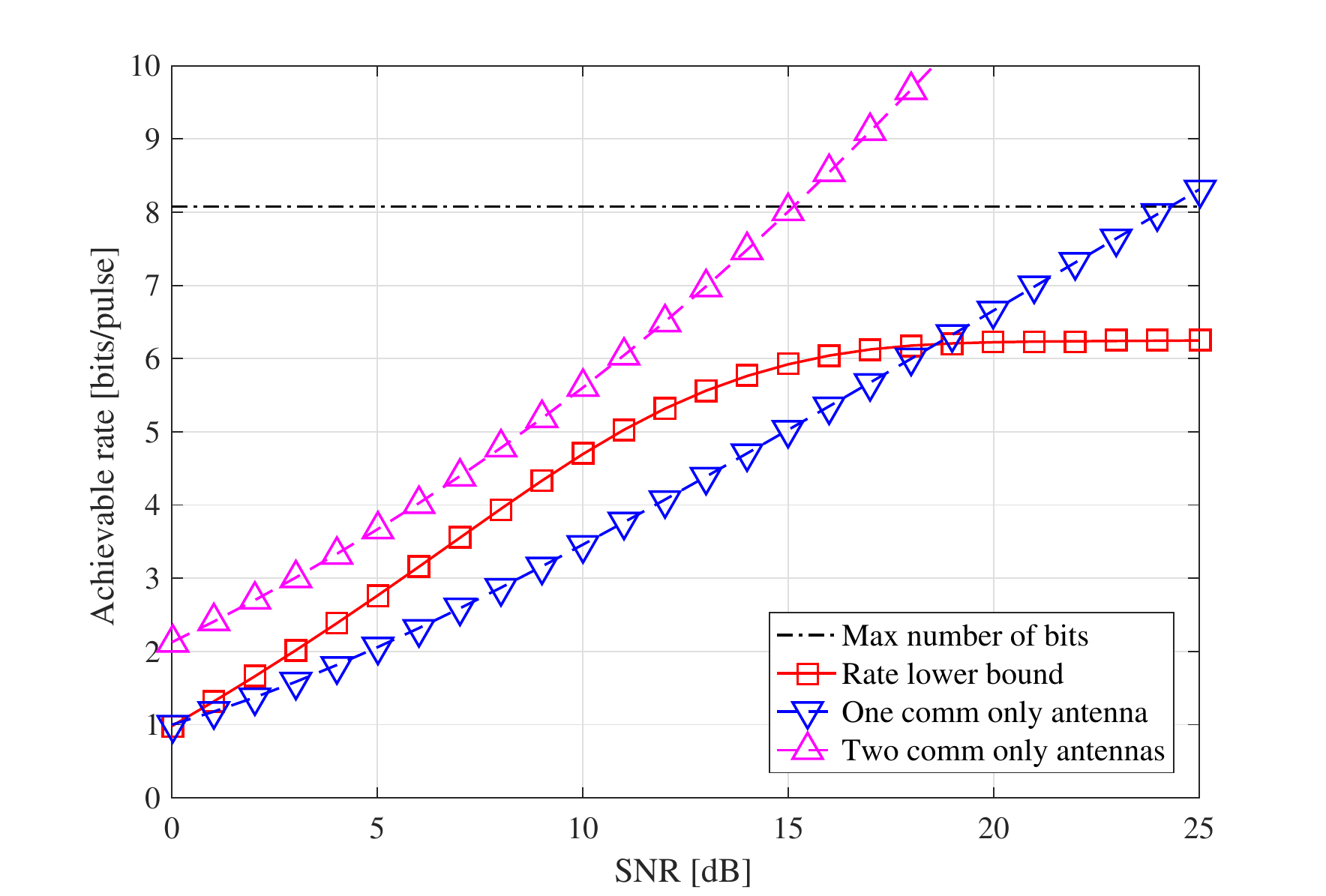} 
		\caption{Achievable rate comparison, spatial decay channel.}
		\label{fig:MultiBand1}
		\vspace{-0.4cm}
	\end{figure}
	\begin{figure}
		\centering
		\includefig{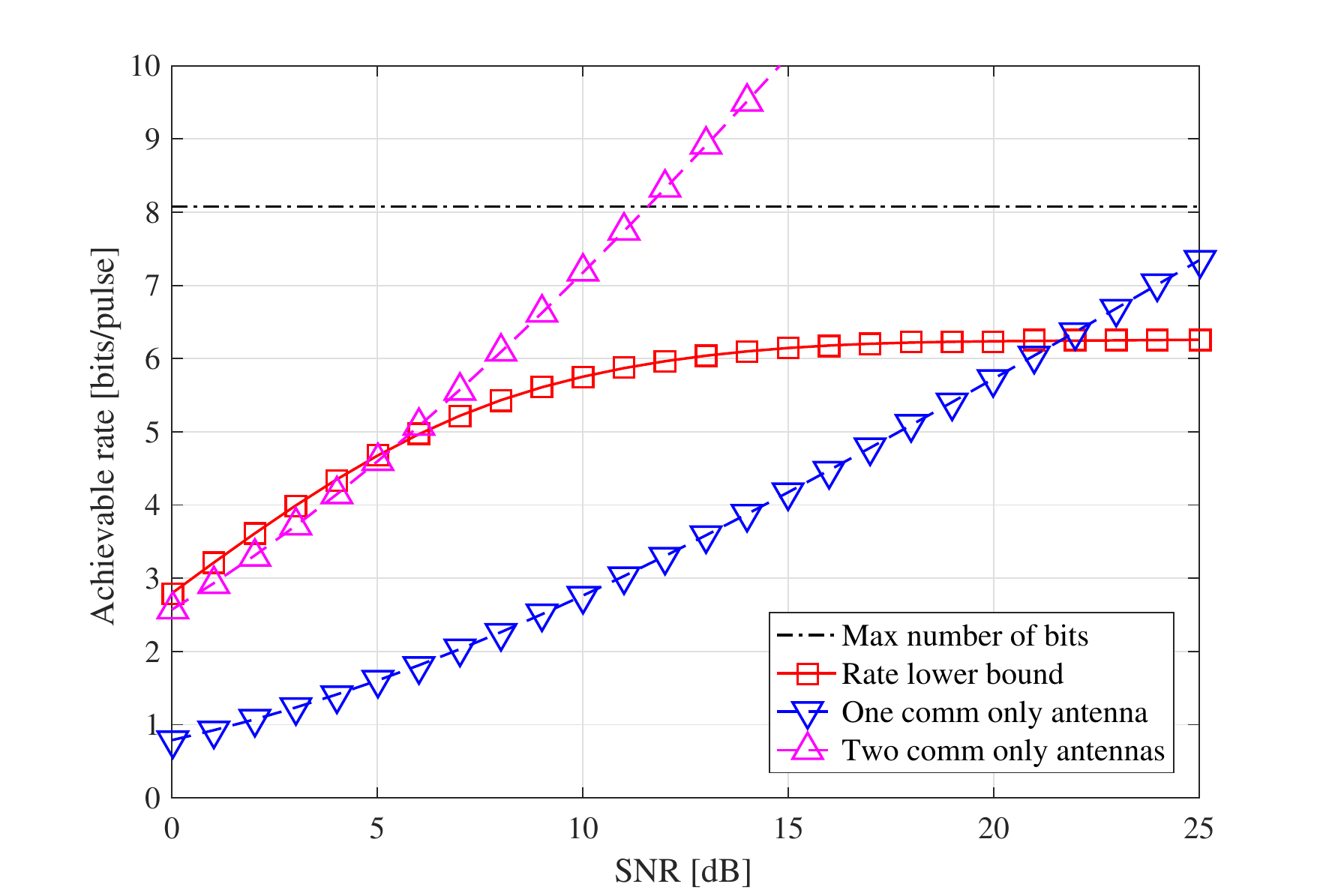} 
		\caption{Achievable rate comparison, Rayleigh fading channel.}
		\label{fig:MultiBand2}
		\vspace{-0.4cm}
	\end{figure}	
	
	\vspace{-0.2cm}
	\subsection{Decoding Strategies}
	\label{subsec:sim_decoder}
	\vspace{-0.1cm}  
	We now evaluate the BER performance of the reduced complexity decoders proposed in Subsection~\ref{subsec:detection}. To that aim, we set the number of transmit and receive antennas to $L_{\rm R} = 6$ and $L_{\rm C} = 4$, respectively, and the channel matrix $\bm H$ is randomized as a zero-mean proper complex Gaussian matrix with i.i.d unit variance entries. The number of available frequencies is $M = 7$, and the beam is directed towards  $\theta = 0$. Here, $L_K = 3$ antenna elements use each frequency. The duration of the pulse is $T_p = 1$ $\mu$s, and the sampling rate used is $\frac{1}{T_s} = M\Delta f$. The number of channel outputs corresponding to each pulse is $L_T =  T_p \cdot  M\Delta f = 70$. 
	The number of bits conveyed by frequency and spatial selections are $\lfloor \log_2\left| \mho \right| \rfloor = 4$ and $\lfloor \log_2\left| \mathcal{P} \right| \rfloor = 4$, respectively.

	We compare the BER performance of the proposed decoders, including the optimal ML decoder \eqref{eq:ML}, denoted 'ML Decoder', and the low complexity decoders proposed in Subsection \ref{subsec:detection} with non-iterative (Alg. \ref{alg:noniter}) and iterative settings (Alg. \ref{alg:iter}). 
	In Alg. \ref{alg:noniter}, we apply both an ML spatial decoder \eqref{eq:ML_spatial} as well as the sub-optimal sequential method with exhaustive search for \eqref{eq:leastsquare} to recover the antenna selection vectors $\bm p_k$, denoted by  'NonIter + ML' and 'NonIter + Greedy', respectively. In Alg. \ref{alg:iter}, we test two approaches: One uses ML for both spatial and spectrum decoders, i.e. \eqref{eq:ML_spatial} and \eqref{eq:ML_frequency}, denoted by 'Iter + ML'; the other one, denoted by 'Iter + Greedy', uses greedy methods, i.e., \eqref{eq:leastsquare} and \eqref{eq:greedy_frequency} to recover the antenna selection vectors $\bm p_k$ and frequencies, respectively. In both iterative algorithms, i.e., 'Iter + ML' and 'Iter + Greedy', the maximum numbers of iterations is $i_{\max}= 10$. The initial estimate of the matrix $\widehat{\bm A}$ is computed via \eqref{eq:AhatFFT}.
	
	In Fig. \ref{fig:BER} we depict the BER  performance  of  these decoders versus SNR, $1/\sigma^2$, averaged over $10^6$ trials. 
	As expected, the computationally complex optimal ML decoder achieves the lowest BER values. Our proposed sub-optimal decoders achieve a performance which scales similarly as the ML decoder with respect to SNR. In particular, the iterative and non-iterative decoders both achieve BER of  $10^{-4}$ at SNR around {-9} dB when combined with ML estimation, while the global ML decoder achieves the same BER at -10 dB, namely, an SNR gap of 1 dB. 
	The corresponding SNR gap of the greedy sequential decoders is 3 dB. In particular, when using the greedy methods, it is observed in Fig. \ref{fig:BER} that estimation refinement using Alg. \ref{alg:iter} does not necessarily improve the accuracy over the initial estimation in Alg. \ref{alg:noniter}. These results indicate that the proposed low complexity decoders are capable of achieving performance  comparable to the  ML decoder while substantially reducing the computational burden at the communication receiver. 
	\begin{figure}
		\centering
		\includefig{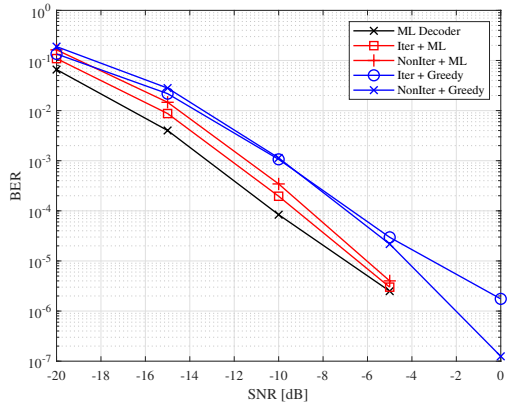} 
		\caption{Bit error rates of the proposed decoders.}
		\label{fig:BER}
		\vspace{-0.4cm}
	\end{figure}
	
	\vspace{-0.2cm}
	\subsection{Codebook Comparison}
	\label{subsec:sim_code}
	\vspace{-0.1cm} 
	Here, we numerically study the codebook design proposed in Subsection~\ref{subsec:Codebook}, and evaluate the impact of the designed codewords on the decoding BER as well as the radar performance. The number of antennas is set to $L_{\rm R} = 8$. Since the codebook does not affect the decoding procedure of the frequency indices, we assume that the transmitted frequencies are already recovered without errors. The remaining settings are the same as those used in the previous study.
	 
	We first evaluate the approximate design criterion minimizing \eqref{eq:dist_vec}, compared to the desired objective \eqref{eq:dist_H}. The numerically computed distances \eqref{eq:dist_vec} and \eqref{eq:dist_H}, denoted 'Dist' and 'H-Dist', respectively, are depicted  in  Fig. \ref{fig:Hdist} for  $L_{\rm C} = 4$. Observing Fig. \ref{fig:Hdist}, we note an approximate monotonic relationship between two distances, which indicates that designing the codewords to minimize \eqref{eq:dist_vec} also reduces the desired objective \eqref{eq:dist_H} proportionally. It is emphasized that when the number of receive antennas $L_{\rm C}$ increases, the monotonicity becomes more distinct. This can be explained since the channel matrix ${\bm H}$ here is Gaussian with i.i.d. entries. Such matrices are known to asymptotically preserve the norm of a projected vector \cite{Eldar2012}, thus \eqref{eq:dist_vec} and \eqref{eq:dist_H} become equivalent. To avoid cluttering, we only present the results for $L_{\rm C} = 4$.
	Comparison between iterative methods and their non-iterative counterparts indicate that iteratively updating improves accuracy of decoders, while the improvement is not significant. 
	\begin{figure}
		\centering
		\includefig{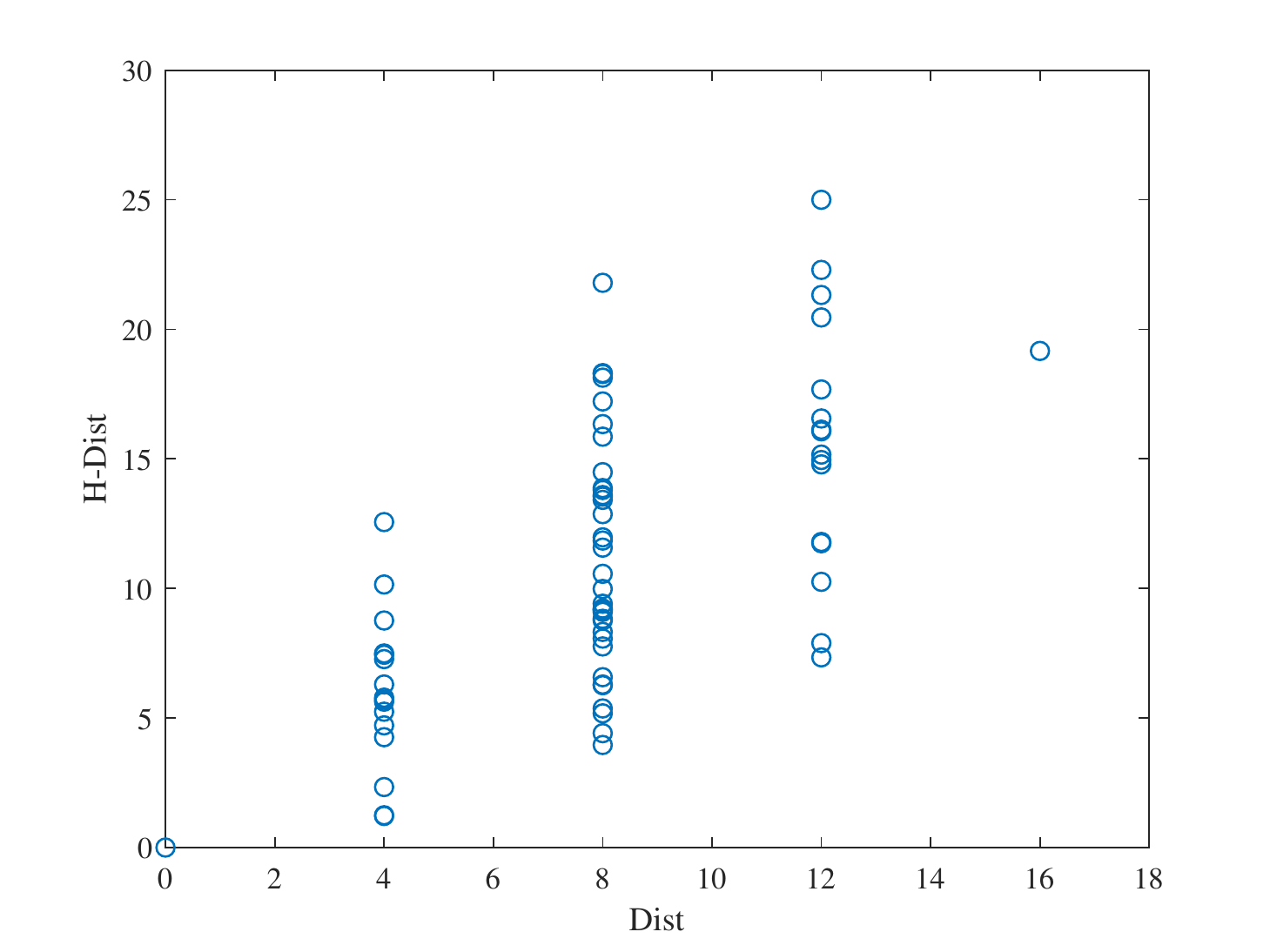} 
		\caption{H-Dist vs. Dist for $L_{\rm C} = 4$. Each circle represents a codeword.}
		\label{fig:Hdist}
		\vspace{-0.4cm}
	\end{figure}

	We next use the objective  \eqref{eq:dist_vec} to design a codebook. 
	After computing distance matrix $\bm R$ in \eqref{eq:Rmatrix}, we use the PCA algorithm to reduce the dimensions of the original codewords, and generate candidate codewords $\widetilde{\bm p} \in \mathbb{R}^{L_{\rm D}}$. The intrinsic dimension of the codewords $\bm p_k$ is estimated as $L_{\rm D} = 7$ here. Given $N_b = 2^1$, $2^3$, $2^5$, the k-means method is applied to cluster the candidates $\widetilde{\bm p}$ into $N_b$ classes. The candidate that is closest to the class center is selected as the final codeword. With these final codewords, we test the BER of the 'NonIter + ML' decoder \eqref{eq:ML_spatial} and depict the results in Fig. \ref{fig:Nb}. As expected, as $N_b$ grows, thus more different messages are conveyed, the overall BER performance is degraded. It is noted that while using smaller $N_b$ values decreases the BER as well as the decoding complexity, it also reduces the data rate, as less bits are conveyed in each symbol.  
	
	\begin{figure}
		\centering
		\includefig{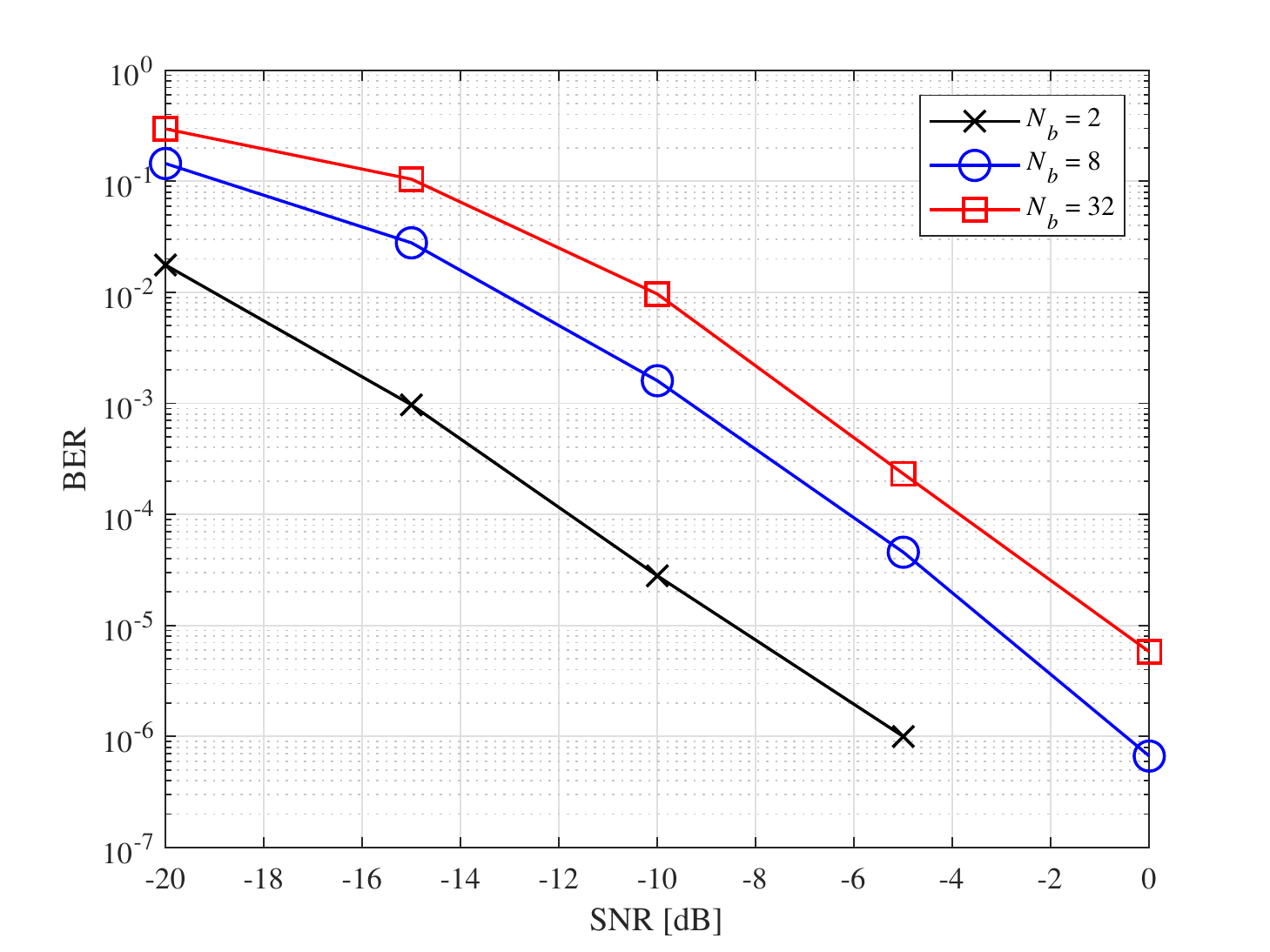} 
		\caption{BERs of the ML spatial decoder \eqref{eq:ML_spatial} for different codebook sizes.}
		\label{fig:Nb}
		\vspace{-0.4cm}
	\end{figure}
	
	Finally, we evaluate the impact of the codebook design on radar performance. In particular, we consider range-Doppler reconstruction and angle estimate of targets being observed, using hit rate and the root mean squared error (RMSE) as performance metrics, respectively. A hit is proclaimed if the range-Doppler parameter of a scattering point is successfully recovered. The RMSE of the target angle is defined as  $\sqrt{{\rm E}[ (\vartheta_s - {\hat{\vartheta}_s})^2 ]}$, where $\vartheta_s$ and $\hat{\vartheta}_s$ denote true angle and estimated one for the $s$-th target, respectively. The number of radar pulses is set to $N = 32$ and is directed to  $\theta = 0$. There are $S = 4$ radar targets inside the beam $\vartheta_s \in \varTheta := \theta+\left[ -\frac{\pi}{2L_{\rm R}},\frac{\pi}{2L_{\rm R}}\right]$ with scattering intensities set to 1. The numerical performance is averaged over 100 Monte Carlo trials. In each trial, the range-Doppler parameters of every target are randomly chosen from the grid points (grid points are explained in \cite[Sec. IV]{Huang2019b}), and the angles are randomly set within the beam $\varTheta$. We define the SNR of the radar returns as $1/\kappa^2$, where $\kappa^2$ is the variance of the additive i.i.d. zero-mean proper-complex Gaussian noise; see \cite[Sec. VII]{Huang2019b}. The algorithm used for radar signal processing is detailed in \cite[Algorithm 1]{Huang2019b}, where Lasso is applied to solve the compressed sensing problem. 
	The resultant range-Doppler reconstruction hit rates and angle estimation performance with  the aforementioned codebooks are depicted in Figs. \ref{fig:full_noisy_hitrate} and \ref{fig:full_noisy_angleerror}, respectively. Observing these two figures, we note that decreasing the codebook size has only a minimal effect on the range-Doppler and angle estimates of radar targets. This indicates  that the proposed codebook reduction method can be used to facilitate the decoding complexity by limiting the number of codewords at the cost 
	of $\log_2 N_b$  less bits  conveyed in each symbol with hardly any impact on estimation performance of radar targets.
	
	\begin{figure}
		\centering
		\includegraphics[width=2.5in]{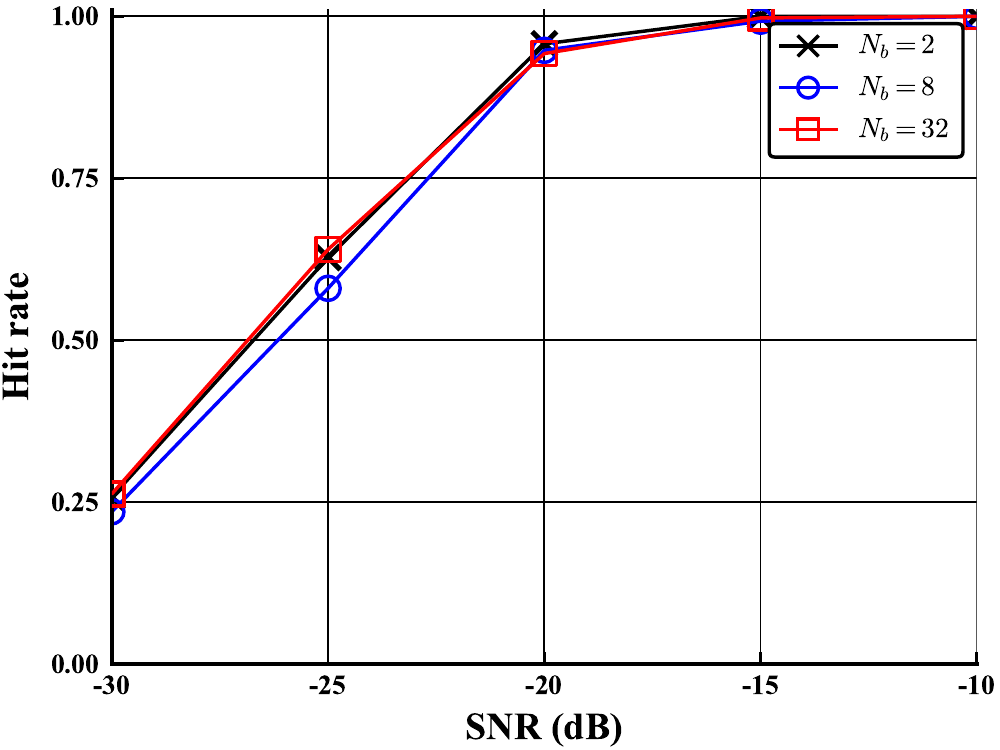}
		\vspace{-0.2cm}
		\caption{Range-Doppler recovery  versus SNR.}
		\label{fig:full_noisy_hitrate}
	\end{figure}
	\begin{figure}
		\centering
		\includegraphics[width=2.5in]{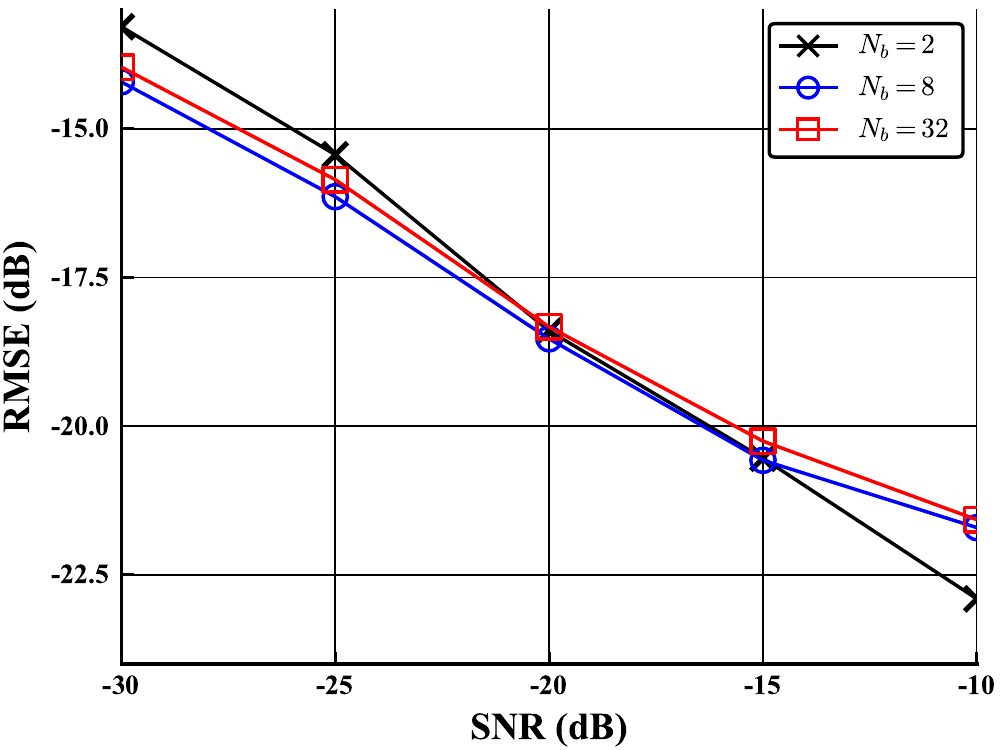}
		\vspace{-0.2cm}
		\caption{Angle estimation performance versus SNR.}
		\label{fig:full_noisy_angleerror}
	\end{figure}	

	\vspace{-0.2cm}
	\section{Conclusions}
	\label{sec:conclusion}
	\vspace{-0.1cm}
	In this paper, we proposed \ac{majorcom} - a DFRC system which combines frequency and spatial agility.   \ac{majorcom} exploits an inherent randomness in the radar scheme to convey information to a remote receiver using index modulation. In particular, the ability of \ac{majorcom} to convey digital messages is a natural byproduct of its radar scheme, and thus does not induce any coexistence and mutual interference issues, unlike most previously proposed DFRC methods.   The achievable rate of the proposed communications scheme was shown to be comparable to that obtained with dedicated communication waveforms without interfering with the radar functionality. To handle the increased decoding complexity of this scheme, a low complexity receiver and codebook design approach were proposed.  Simulation results demonstrate that \ac{majorcom} exhibits excellent communication performance, and that  the proposed low complexity techniques allow to efficiently balance computational burden and communication reliability. 
	
	\vspace{-0.2cm}
	\appendix[Proof of Proposition \ref{prop:R_diag}]
	\label{sec:app}
	\vspace{-0.1cm}
	The symmetry and zero main diagonal of $\bm R$ follow directly from its definition  \eqref{eq:Rmatrix}. We thus only prove that each row of $\bm R$ is a permutation of its first row.
	
	For each codeword $i$, there exists an $L_{\rm R}\times L_{\rm R}$ permutation matrix $\bm\Sigma_i$ 
	such that $\bm\Sigma_i\bm p_k^{(0)}=\bm p_k^{(i)}$, $k=0,\ldots,K-1$. 
	This permutation matrix is not unique: two permutations $\bm\Sigma$, $\tilde{\bm\Sigma}$ 
	induce the same codeword 
	(i.e., $\bm\Sigma\bm p_k^{(0)}=\tilde{\bm\Sigma}\bm p_k^{(0)}$ for all $k$) 
	if and only if  $\bm\Sigma^{-1}\tilde{\bm\Sigma} \bm p_k^{(0)} = \bm p_k^{(0)}$, for all $k$. 
	For convenience, denote by $\mathcal{G}$ the set of all permutation matrices that fix $\bm p_k^{(0)}$ for all $k$. 
	Choose for each $i$ a permutation matrix $\bm\Sigma_i$ inducing codeword $i$. 
	The $i$-th row of $\bm R$ consists of elements
	\begin{equation}
	[\bm R]_{i,j}=\sum_{k=0}^{K-1}\left\|\bm\Sigma_i\bm p_k^{(0)}-\bm\Sigma_j\bm p_k^{(0)}\right\|^2.
	\end{equation}
	Since permutation matrices are orthogonal, this is equal to 
	\begin{equation}
	\sum_{k=0}^{K-1}\left\|{\bm p}_k^{(0)}-\bm\Sigma_i^{-1}\bm\Sigma_j\bm p_k^{(0)}\right\|^2.
	\end{equation}
	Denote by $\text{code}^i_j$ the codeword induced by $\bm\Sigma_i^{-1}\bm\Sigma_j$. 
	Then, 
	\vspace{-0.1cm}
	\begin{equation}\label{eqn:row_code_j}
	[\bm R]_{i,j}=\sum_{k=0}^{K-1}\left\|\bm p_k^{(0)}-\bm p_k^{(\text{code}^i_j)}\right\|^2.
	\vspace{-0.1cm}
	\end{equation}
	For $j\ne j'$, we note that $\bm\Sigma_i^{-1}\bm\Sigma_j$ and $\bm\Sigma_i^{-1}\bm\Sigma_{j'}$ 
	induce different codewords since
	\vspace{-0.1cm}
	\begin{equation}
	(\bm\Sigma_i^{-1}\bm\Sigma_j)^{-1}\bm\Sigma_i^{-1}\bm\Sigma_{j'}=\bm\Sigma_j^{-1}\bm\Sigma_{j'}\not\in \mathcal{G}.
	\vspace{-0.1cm}
	\end{equation}
	Thus, as $j$ runs through all the codewords, both $\text{code}^0_j$ and $\text{code}^i_j$ run through all the codewords. 
	By (\ref{eqn:row_code_j}) this implies that the $i$-th row of $\bm R$ is a permutation of the first row of $\bm R$. 
	$\qed$
	
	%
	

	\ifCLASSOPTIONcaptionsoff
	\newpage
	\fi

	
	
	%
	\vspace{-0.2cm}
	\bibliographystyle{IEEEtran}
	\bibliography{ref2}

	
	

\end{document}